\documentclass{article} 
\usepackage{fullpage}
\usepackage{hyperref}

\usepackage{paper,bm,graphicx,url}
\usepackage[boxed]{algorithm2e}

\title{\LARGE{Private Approximations of the 2nd-Moment Matrix Using Existing Techniques in Linear Regression}}

\author{Or Sheffet\\Center for Research on Computation and Society\\Harvard University\\Cambridge, MA\\\texttt{osheffet@seas.harvard.edu}}

\newcommand{\gauss}[1]{\mathcal{N}\left(#1\right)}
\newcommand{\T}{{^{\mathsf{T}}}}
\renewcommand{\vec}[1]{\bm{#1}}
\newcommand{\PDF}{\ensuremath{\mathsf{PDF}}}
\newcommand{\lndelta}{\ensuremath{\ln(\tfrac 1 \delta)}}
\newcommand{\cut}[1]{}
\newcommand{\term}{\ensuremath{\left(\sqrt d + \sqrt{2\ln(\tfrac 2 \delta)}\right)}}
\newcommand{\W}{{\mathcal{W}}}
\renewcommand{\N}{{\mathcal{N}}}
\newcommand{\hvec}[1]{\ensuremath{\widehat{\vec{#1}}}}
\newcommand{\tvec}[1]{\ensuremath{\widetilde{\vec{#1}}}}
\newcommand{\trace}{\mathrm{tr}}
\newcommand{\mpi}{{\dagger}}
\newcommand{\svmin}{{\sigma_{\min}}}
\newcommand{\svmax}{{\sigma_{\max}}}
\newcommand{\minimize}{\ensuremath{\textrm{minimize}}}
\newcommand{\mypic}[1]{\includegraphics[width=0.9\textwidth,keepaspectratio=true]{#1}}
\newcommand{\mypicinmaintext}[1]{\includegraphics[width=0.75\textwidth,keepaspectratio=true]{#1}}
\begin{document}

\ifdefined\NIPS
\newcommand{\myvspace}[1]{\vspace{#1}}
\newcommand{\myparagraph}[1] {\myvspace{1pt}\noindent\textbf{#1}}
\bibliographystyle{unsrt}
\else
\newcommand{\myvspace}[1]{}
\newcommand{\myparagraph}[1]{\paragraph{#1}}
\bibliographystyle{alpha}
\fi

\newcommand{\fix}{\marginpar{FIX}}
\newcommand{\new}{\marginpar{NEW}}

\maketitle
\myvspace{-25pt}
\begin{abstract}
\myvspace{-13pt}
We introduce three differentially-private algorithms that approximate the 2nd-moment matrix of the data. These algorithm, which in contrast to existing algorithms output positive-definite matrices, correspond to existing techniques in linear regression literature. Specifically, we discuss the following three techniques. (i) For Ridge Regression, we propose setting the regularization coefficient so that by approximating the solution using Johnson-Lindenstrauss transform we preserve privacy. (ii) We show that adding a small batch of random samples to our data preserves differential privacy. (iii) We show that sampling the 2nd-moment matrix from a Bayesian posterior inverse-Wishart distribution is differentially private provided the prior is set correctly. We also evaluate our techniques experimentally and compare them to the existing ``Analyze Gauss'' algorithm of Dwork et al~\cite{DworkTTZ14}.
\end{abstract}
\myvspace{-20pt}
\section{Introduction}
\label{sec:intro}
\myvspace{-8pt}
Differentially private algorithms~\cite{DworkMNS06, DworkKMMN06} are data analysis algorithms that give a strong guarantee of privacy, roughly stated as: by adding to or removing from the data a single datapoint we do not significantly change the probability of any outcome of the algorithm. 
The focus of this paper is on differentially private approximations of the $2$nd-moment matrix of the data --- given a dataset $D\in \mathbb{R}^{n\times d}$, its \emph{$2$nd-moment matrix} (also referred to as the \emph{Gram} matrix of data or the \emph{scatter matrix} if the mean of $D$ is $\vec 0$) is the matrix $D\T D$ --- and the uses of such approximations in linear regression.
Indeed, since the $2$nd-moment matrix of the data plays a major role in many data-analysis techniques, we already have differentially private algorithms that approximate the $2$nd-moment matrix~\cite{DworkTTZ14} for the purpose of approximating the PCA, techniques for approximating the rank-$k$ PCA of the data directly~\cite{HardtR12, Hardt13, KapralovT13}, or  differentially private algorithms for linear regressions~\cite{ChaudhuriMS11, KiferST12, ThakurtaS13, BassilyST14}.

However, existing techniques for differentially private linear regression suffer from the drawback that they approximate a single regression. That is, they assume that each datapoint is composed of a vector of features $\vec x$ and a label $y$ and find the best linear combination of the features that predicts $y$. Yet, given a dataset $D$ with $d$ attributes we are free to pick any single attribute as a label, and any subset of the remaining attributes as features. Therefore, a database with $d$ attributes yields $\exp(d)$ potential linear regression problems; and running these algorithms for each linear regression problem separately simply introduces far too much random noise.\footnote{Indeed, Ullman~\cite{Ullman15} have devised a solution to this problem, but this solution works in the more-cumbersome online model and requires exponential running-time for the curator; whereas our techniques follow the more efficient offline approach.}

In contrast, the differentially private techniques that approximate the 2nd-moment matrix of the data, such as the Analyze Gauss paper of Dwork et al~\cite{DworkTTZ14}, allow us to run as many regressions on the data as we want. Yet, to the best of our knowledge, they have never been analyzed for the purpose of linear regression. Furthermore, the Analyze Gauss algorithm suffers from the drawback that it does not necessarily output a positive-definite matrix. This, as discussed in~\cite{XiKI11} and as we show in our experiments, can be very detrimental --- even if we do project the output back onto the set of positive definite matrices. And though the focus of this work is on linear regression, one can postulate additional reasons why releasing a positive definite matrix is of importance, such as using the output as a kernel matrix or doing statistical inference on top of the linear regression.

\myparagraph{Our Contribution.} In this work, we give three differentially private techniques for approximating the 2nd-moment matrix of the data that output a positive-definite matrix. We analyze their utility, both theoretically and empirically, and more importantly --- show how they correspond to \emph{existing techniques in linear regression}. And so we contribute to an increasing line of works~\cite{BlockiBDS12, VadhanZ15, WangFS15} that shows that differential privacy may rise from existing techniques, provided parameters are set properly. We also compare our algorithms to the existing Analyze Gauss technique.\\
(Some notation before we introduce our techniques. We assume the data is a matrix $A\in \mathbb{R}^{n\times d}$ with $n$ sample points in $d$ dimensions. For the ease of exposition, we focus on a single regression problem, given by $A=[X;\vec y]$ --- i.e., the label is the $d$-th column and the features are the remaining $p=d-1$ columns. We use $\svmin(A)$ to denote the least singular value of $A$.)

\emph{1. The Johnson-Lindenstrauss Transform and Ridge Regression.} Blocki et al~\cite{BlockiBDS12} have shown that projecting the data using a Gaussian Johnson-Lindenstrauss transform preserves privacy if $\svmin(A)$ is sufficiently large and it has been applied for linear regression~\cite{U14}. Our first result improves on the analysis of Blocki et al  and uses a smaller bound on $\svmin(A)$ (shaving off a factor of $\log(r)$ with $r$ denoting the number of rows in the JL transform). This result implies that when $\svmin(A)$ is large we can project the data using the JL-transform and output the $2$nd-moment matrix of the projected data and preserve privacy. Furthermore, it is also known~\cite{Sarlos06} that the JL-transform gives a good approximation for linear regression problems. However, this is somewhat contradictory to our intuition: for datasets where $\vec y$ is well approximated by a linear combination of $X$, the least singular value should be small (as $A$'s stretch along the direction $\left(\vec\beta, -1\right)\T$ is small). That is why we artificially increase the singular values of $A$ by appending it with a matrix $w\cdot I_{d\times d}$. It turns out that this corresponds to approximating the solution of the \emph{Ridge regression} problem~\cite{Tikhonov63,HoerlK70}, the linear regression problem with $l_2$-regularization  --- the problem of finding $\vec\beta^R = \arg\min_{\vec\beta} \sum_i \| y_i - \vec\beta \cdot \vec x_i\|^2 + w^2 \|\vec\beta\|^2$. Literature suggests many approaches~\cite{HastieTF09} to determining the penalty coefficient $w^2$, approaches that are based on the data itself and on minimizing risk. Here we give a fundamentally different approach~--- set $w$ as to preserve $(\epsilon,\delta)$-differential privacy. Details, utility analysis and experiments regarding this approach appear in Section~\ref{sec:ridge_regression}.

\emph{2. Additive Wishart noise.} Whereas the Analyze Gauss algorithm adds Gaussian noise to $A\T A$, here we show that we can sample a positive definite matrix $W$ from a suitably chosen Wishart distribution $\W_d(V,k)$, and output $A\T A + W$. This in turn corresponds to appending $A$ with $k$ i.i.d samples from a multivariate Gaussian $\gauss{\vec 0_d, V}$. One is able to view this too as an extension of Ridge regression, where instead of appending $A$ with $d$ fixed examples, we append $A$ with $k\approx d+O(1/\epsilon^2)$ random examples.\footnote{Though it is also tempting to think of this technique as running Bayesian regression with random prior, this analogy does not fully carry through as we discuss later.} Note, as opposed to Analyze Gauss~\cite{DworkTTZ14}, where the noise has $0$-mean, here the expected value of the noise is $k V$. This yields a useful way of post-processing the output: $A\T A + W - k V$. Details, theorems and experiments with additive Wishart noise appear in Section~\ref{sec:additive_Wishart}.

\emph{3. Sampling from an inverse-Wishart distribution.} The Bayesian approach for estimating the 2nd-moment matrix of the data assumes that the $n$ sample points are sampled i.i.d from some $\gauss{\vec 0_d, V}$ for some unknown $V$, where we have a prior distribution on $V$. Each sample point causes us to update our belief on $V$ which results in a posterior distribution on $V$. Though often one just outputs the MAP of the posterior belief (the mean of the posterior distribution), it is also common to output a sample drawn randomly from the posterior distribution. We show that if one uses the inverse-Wishart distribution as a prior (which is common, as the inverse-Wishart distribution is a conjugate prior), then sampling from the posterior is $(\epsilon,\delta)$-diffrentially private, provided the prior is spread enough. This gives rise to our third approach of approximating $A\T A$ --- sampling from a suitable inverse Wishart distribution. We comment that the idea that existing techniques in Bayesian analysis, and specifically sampling from the posterior distribution, are differentially-private on their own was originally introduced in the beautiful and elegant work of Vadhan and Zheng~\cite{VadhanZ15}. But whereas their work focuses on estimating the mean of the sample, we focus on estimating the variances/2nd-moment. Details, theorems and experiments on sampling from the inverse-Wishart distribution appear in Section~\ref{sec:inv_wishart}.

Finally, in Section~\ref{sec:AG_comparison} we compare our algorithms to the Analyze Gauss algorithm. We show that in the simple case where the data is devised by $p$ independent features concatenated with a single linear combination of the features, the Analyze Gauss algorithm, which introduces the least noise out of all algorithms, is clearly the best algorithm once $n$ is sufficiently large. However, when the data contains multiple such regressions and therefore has small singular values, the situation is far from being clear cut, and indeed, unless $n$ is extremely large, our algorithms achieve smaller errors than the Analyze Gauss baseline. We comment that our experiments should be viewed solely as a proof-of-conecpt. They are only preliminary, and much more experimentation is needed to fully evaluate the benefits of the various algorithms. 

\myparagraph{Our proof technique.} Before continuing to preliminaries and the formal details of our algorithms, we give an overview of the proof technique. (All of the proofs are deferred to Appendix~\ref{sec:privacy_analysis}.) To prove that each algorithm preserves $(\epsilon,\delta)$-differential privacy we state and prove $3$ corresponding theorems, whose proofs follow the same high-level approach. As mentioned above, one theorem improves on a theorem of Blocki et al~\cite{BlockiBDS12}, who were the first to show that the JL-transform is differentially private. Blocki et al observed that by projecting the data using a $(r\times n)$-matrix of i.i.d normal Gaussians, we effectively repeat the same one-dimensional projection $r$ independent times. So they proved that each one-dimensional projection is $(\epsilon,\delta)$-differentially private, and to show the entire projection preserves privacy they used the off-the-shelf composition of Dwork et al~\cite{DRV10}, getting a bound that depends on $O(\sqrt{r}\log(r))$. In order to derive a bound depending only on $O(\sqrt r)$, we do not use the composition theorem of~\cite{DRV10} but rather study the specific $r$-fold composition of the projection. As a result, we cannot follow the approach of Blocki et al.

To show that a one-dimensional projection is $(\epsilon,\delta)$-differentially private, Blocki et al compared the $\PDF$s of two multivariate Gaussians. The $\PDF$ of a multivariate Gaussian is given by the multiplication of two terms: the first depends on the determinant of the variance, and the second depends on some exponent (see exact definition in Section~\ref{sec:preliminaries}). Blocki et al compared the ratio of each of the terms and showed that w.h.p each term's ratio is bounded by $e^{\epsilon/2}$. Unfortunately, following the same approach of Blocki et al yields a bound of $e^{r\epsilon/2}$ for each of the terms and an overall bound that depends on $O(r)$. Instead, we observe that the contributions of the determinant term and the exponent term to the ratio of the $\PDF$s are of opposite signs. So we use the Matrix Determinant Lemma and the Sherman-Morrison Lemma (see Theorem~\ref{thm:Sherman_Morrison}) to combine both terms into a single exponent term, and bound its size using the Johnson-Lindenstrauss transform (or rather, tight bounds on the $\chi^2$-distribution). The main lemma we use in our analysis is detailed in Lemma~\ref{lem:JL_of_Wishart_and_invWishart}. This lemma, in addition to giving tight bounds for the Gaussian JL-transform (mimicking the approach of Dasgupta and Gupta~\cite{DasguptaG03}), also gives a result that might be of independent interest. The standard JL lemma shows that for a $(r\times d)$-matrix $R$ of i.i.d normal Gaussians and any fixed vector $\vec v$ it holds w.h.p that $\vec v\T \vec v \in (1\pm \eta)\vec v\T (\tfrac 1 r R\T R)\vec v$ provided $r =O(\eta^{-2})$. In Lemma~\ref{lem:JL_of_Wishart_and_invWishart} we also show that for any fixed $\vec v$ we have w.h.p. that $\vec v\T \vec v \in ({1\pm\eta}) \vec v\T (\tfrac 1 {r-d}R\T R)^{-1} \vec v$ provided $r = d + O(\eta^{-2})$. \footnote{To the best of our knowledge, for a general JLT, this is known to hold only when $r = O(d\cdot\eta^{-2})$ and the transform preserves the lengths of all vectors in the $\mathbb{R}^d$ space, see~\cite{Sarlos06} Corollary 11.}
\myvspace{-15pt}
\section{Preliminaries and Notation}
\label{sec:preliminaries}
\myvspace{-8pt}
\myparagraph{Notation.} Throughout this paper, we use $lower$-case letters to denote scalars; $\pmb{bold}$ characters to denote vectors; and UPPER-case letters to denote matrices. The $l$-dimensional all zero vector is denoted $\vec 0_l$, and the $(l\times m)$-matrix of all zeros is denoted $0_{l\times m}$. The $l$-dimensional identity matrix is denoted $I_{l\times l}$. For two matrices $M,N$ with the same number of row we use $[M;N]$ to denote the concatenation of $M$ and $N$.
We use $\epsilon,\delta$ to denote the privacy parameters. For a given matrix, $\|M\|$ denotes the spectral norm ($=\svmax(M)$) and $\|M\|_F$ denotes the Frobenious norm $(\sum_{j,k} M_{j,k}^2)^{1/2}$; and use $\svmax(M)$ and $\svmin(M)$ to denote its largest and smallest singular value resp.


\myparagraph{The Gaussian Distribution and Related Distributions.} 
We denote by $Lap(\sigma)$ the Laplace distribution whose mean is $0$ and variance is $2\sigma^2$. A univariate Gaussian $\gauss{\mu, \sigma^2}$ denotes the Gaussian distribution whose mean is $\mu$ and variance $\sigma^2$. Standard concentration bounds on Gaussians give that $\Pr[ x> \mu+\sigma\sqrt{\ln(1/\nu)}] < \nu$. A multivariate Gaussian $\gauss{\vec\mu,\Sigma}$ for some positive semi-definite $\Sigma$ denotes the multivariate Gaussian distribution where the mean  of the $j$-th coordinate is the $\mu_j$ and the co-variance between coordinates $j$ and $k$ is $\Sigma_{j,k}$. The $\PDF$ of such Gaussian is defined only on the subspace $colspan(\Sigma)$, where for every $x\in colspan(\Sigma)$ we have $\PDF(\vec x) = \left((2\pi)^{rank(\Sigma)}\cdot \tilde\det(\Sigma)\right)^{-1/2}\exp\left(-\tfrac 1 2 (\vec x-\vec \mu)^\T \Sigma^\mpi (\vec x-\vec\mu)\right)$ and $\tilde\det(\Sigma)$ is the multiplication of all non-zero singular values of $\Sigma$. We will repeatedly use the rules regarding linear operations on Gaussians. That is, for any scalar $c$, it holds that $c\gauss{\mu,\sigma^2} = \gauss{c\cdot\mu,c^2\sigma^2}$. For any matrix $C$ it holds that $C\cdot \gauss{\vec \mu, \Sigma} = \gauss{C\vec\mu, C \Sigma C^\T}$. 

The $\chi^2_k$-distribution, where $k$ is referred to as the degrees of freedom of the distribution, is the distribution over the $l_2$-norm of the sum of $k$ independent normal Gaussians. That is, given $X_1, \ldots, X_k \sim \gauss{0,1}$ it holds that $\zeta \stackrel{\rm def}= (X_1, X_2, \ldots, X_k) \sim\gauss{\vec 0_k, I_{k\times k}}$, and  $\|\zeta\|^2\sim \chi^2_k$. Standard tail bounds on the $\chi^2$-distribution give that for any $\nu\in(0,\tfrac 1 2)$ we have $\Pr_{x\sim \chi^2_k}[ x \in \left(\sqrt{k}\pm \sqrt{2\ln(2/\nu)}\right)^2 ]\geq 1-\nu$. (We present them in Section~\ref{sec:supporting_lemmas} for completeness.) The Wishart-distribution $\W_d(V,m)$ is the multivariate extension of the $\chi^2$-distribution. It describes the scatter matrix of a sample of $m$ i.i.d samples from a multivariate Gaussian $\gauss{\vec 0_d, V}$ and so the support of the distribution is on positive definite matrices. For $m>d-1$ we have that $\PDF_{\W_d(V,m)}(X) \propto \det(V)^{-\frac m 2} \det(X)^{\frac {m-d-1}2}\exp(-\tfrac 1 2 \trace(V^{-1}X)$. The inverse-Wishart distribution $\W^{-1}_d(V,m)$ describes the distribution over positive definite matrices whose inverse is sampled from the Wishart distribution using the inverse of $V$; i.e. $X\sim W^{-1}_d(V,m)$ iff $X^{-1}\sim \W_d(V^{-1}, m)$. For $m>d-1$ it holds that $\PDF_{\W^{-1}_d(V,m)}(X) \propto \det(V)^{\frac m 2}\det(X)^{-\frac{m+d+1}2} \exp(-\tfrac 1 2 \trace( V X^{-1}))$.

\myparagraph{Differential Privacy.} In this work, we deal with input of the form of a $(n\times d)$-matrix with each row bounded by a $l_2$-norm of $B$. Converting $A$ into a linear regression problem, we denote $A$ as the concatenation of the $(n\times p)$-matrix $X$ with the vector $\vec y\in \mathbb{R}^n$ ($A=[X;\vec y]$) where $p=d-1$. This implies we are tying to predict $\vec y$ as a linear combination of the columns of $X$. Two matrices $A$ and $A'$ are called \emph{neighbors} if they differ on a single row. 
\begin{definition}[\cite{DworkMNS06,DworkKMMN06}]
\label{def:privacy}
An algorithm $\textsf{ALG}$ which maps $(n\times d)$-matrices into some range $\mathcal{R}$ is \emph{$(\epsilon,\delta)$-differential privacy} if for all pairs of neighboring inputs $A$ and $A'$ and all subsets $\mathcal{S}\subset\mathcal{R}$ it holds that
$\Pr[\mathsf{ALG}(A) \in \mathcal{S}] \leq e^\epsilon \Pr[\mathsf{ALG}(A') \in \mathcal{S}] + \delta$. When $\delta=0$ we say the algorithm is $\epsilon$-differentially private.
\end{definition}
It was shown in~\cite{DworkMNS06} that for any $f$ where $\|f(A)-f(A')\|_1 \leq \Delta$ then the algorithm that adds Laplace noise $Lap(\tfrac \Delta \epsilon)$ to $f(A)$ is $\epsilon$-differential privacy. It was shown in~\cite{DworkKMMN06} that for any $f$ where $\|f(A)-f(A')\|_2 \leq \Delta$ then adding Laplace noise $\gauss{0, \tfrac {2\Delta^2\ln(2/\delta)}{\epsilon}}$ to $f(A)$ is $(\epsilon,\delta)$-differential privacy. This is precisely the algorithm of Dwork et al in their ``Analyze Gauss'' paper~\cite{DworkTTZ14}. They observed that in our setting, for the function $f(A) = A\T A$ we have that $\|f(A)-f(A')\|_F^2 = B^4$. And so they add i.i.d Gaussian noise to each coordinate of $A\T A$ (forcing the noise to be symmetric, as $A\T A$ is symmetric). We therefore refer to this benchmark as the Analyze Gauss algorithm. In addition, it is known that the composition of two algorithms, each of which is $(\epsilon,\delta)$-differentially private, yields an algorithm which is $(2\epsilon,2\delta)$-differentially private. 
\myvspace{-12pt}
\section{Ridge Regression --- Set the Regularization Coefficient to Preserve Privacy}
\label{sec:ridge_regression}
\myvspace{-8pt}
The standard problem of linear regression, finding $\hvec\beta = \arg\min_{\vec\beta} \|X\vec\beta - \vec y\|^2$, relies on the fact that $X$ is of full-rank. This clearly isn't always the case, and $X\T X$ may be singular or close to singular. To that end, as well as for the purpose of preventing over-fitting, regularization is introduced. One way to regularize the linear regression problem is to introduce a $l_2$-penalty term: finding $\vec\beta^R = \arg\min_{\vec\beta} \|X\vec\beta - \vec y \|^2 + w^2 \|\vec\beta\|^2$. This is known as the \emph{Ridge regression} problem, introduce by~\cite{Tikhonov63, HoerlK70} in the 60s and 70s. Ridge regression has a closed form solution: $\vec\beta^R = (X\T X + w^2 I_{p\times p})X\T y$. The problem of setting $w$ has been well-studied~\cite{HastieTF09} where existing techniques are data-driven, often proposing to set $w$ as to minimize the risk of $\vec\beta^R$.  
Here, we propose a fundamentally different approach to the problem of setting $w$: set it so that we can satisfy $(\epsilon,\delta)$-differential privacy (via the Johnson-Lindenstrauss transform). 

Observe, the Ridge regression problem can be written as: $\minimize \|X\vec\beta - \vec y \|^2 + \|w I_{p\times p} \vec{\beta} - \vec 0_p\|^2$. So, denote $X'$ are the $((n+p)\times p)$-matrix which we get by concatenating $X$ and $wI_{p\times p}$, and denote $\vec y'$ as the concatenation of $\vec y$ with $p$ zeros. Then $\beta^R = \arg\min \|X'\vec\beta - \vec y'\|^2$. Since $p=d-1$ and we denote $A =[X;\vec y]$, we can in fact set $A'$ as the concatenation of $A$ with the $d$-dimensional matrix $wI_{d\times d}$, and we have that $f(\vec\beta) \stackrel{\rm def} = \left\| A' \left(\begin{array}{c}\vec\beta\cr -1\end{array}\right)\right\|^2 = \|X'\vec\beta-\vec y'\|^2 + w^2$. Hence $\vec\beta^R = \arg\min f(\vec\beta)$.
Hence, an approximation of $A'\T A'$ yields an approximation of the Ridge regression problem. One way to approximate $A'\T A'$ is via the Johnson-Lindenstrauss transform, which is known to be differentially private if all the singular values of the given input are sufficiently large~\cite{BlockiBDS12}. And that is precisely why we use $A'$ --- all the singular values of $A'\T A'$ are greater by $w^2$ than the singular values of$A\T A$, and in particular are always $\geq w^2$. Therefore, applying the JLT to $A'$ gives an approximation of $A'\T A'$, and furthermore, due to the work of Sarlos~\cite{Sarlos06} the JLT also approximates the linear regression. The following theorem improves on the original theorem of Blocki et al~\cite{BlockiBDS12}.

\begin{theorem}
\label{thm:JL_body}
Fix $\epsilon>0$ and $\delta \in(0,\tfrac 1 e)$. Fix $B>0$. 
Fix a positive integer $r$ and let $w$ be such that $w^2 ={4B^2}\left(\sqrt{2 r\ln(\tfrac 4 \delta)} + {\ln(\tfrac 4 \delta)}\right)/\epsilon$. Let $A$ be a $(n\times d)$-matrix with $d<r$ and where each row of $A$ has bounded $L_2$-norm of $B$. Given that $\sigma_{\min}(A) \geq w$, the algorithm that picks a $(r\times n)$-matrix $R$ whose entries are i.i.d samples from a normal  distribution $\gauss{0,1}$ and publishes $R\cdot A$ is $(\epsilon,\delta)$-differentially private.
\end{theorem}

This gives rise to our first algorithm. Algorithm~\ref{alg:private_JL} gets as input the parameter $r$ --- the number of rows in our JLT, and chooses the appropriate regularization coefficient $w$. Based on Theorem~\ref{thm:JL_body} and above-mentioned discussion, it is clear that Algorithm~\ref{alg:private_JL} is $(\epsilon,\delta)$-differentially private. Furthermore, based on the work of Sarlos, we can also argue the following.

\begin{theorem}
\label{thm:utility_jl_ridge_regression} [\cite{Sarlos06}, Theorem 12] Fix any $\eta > 0$ and $\nu \in (0,\tfrac 1 2)$. Apply Algorithm~\ref{alg:private_JL} with $r = O(d\log(d)\ln(1/\nu)/\eta^2)$. Then, w.p $\geq 1-\nu$ it holds that $\|\vec\beta^R -\tvec\beta^R\| \leq \frac {\eta}{\sqrt{w^2+\svmin(A\T A)}} f(\vec\beta^R)$.
\end{theorem}

Existing results about the expected distance $\E[\|\vec\beta^R-\hvec\beta\|^2]$ (see~\cite{DhillonFKU13}) can be used together with Theorem~\ref{thm:utility_jl_ridge_regression} to give a bound on $\|\tvec\beta^R-\hvec\beta\|^2$.

In addition to Algorithm~\ref{alg:private_JL}, we can use part of the privacy budget to look at the least singular-value of $A\T A$. If it happens to be the case that $\svmin(A\T A)$ is large, then we can adjust $w$ by decreasing it by the appropriate factor. In fact, one can completely invert the algorithm and, in case $\svmin(A\T A)$ is really large, not only set the regularization coefficient to be any arbitrary non-negative number, but also determine $r$ based on Thm~\ref{thm:JL_body}. Details appear in Algorithm~\ref{alg:private_JL_estimating_w}.

\begin{algorithm}[htb]
\DontPrintSemicolon
\KwIn{A matrix $A \in \mathbb{R}^{n\times d}$ and a bound $B>0$ on the $l_2$-norm of any row in $A$.\\ \qquad Privacy parameters: $\epsilon, \delta > 0$.\\ \qquad Parameter $r$ indicating the number of rows in the resulting matrix.}
Set $w = \sqrt{{4B^2}\left(\sqrt{2 r\ln(\tfrac 4 \delta)} + {\ln(\tfrac 4 \delta)}\right)/\epsilon}$. \;
Set $A'$ as the concatenation of $A$ with $wI_{d\times d}$.\;
Sample a $r\times (n+d)$-matrix $R$ whose entries are i.i.d samples from a normal Gaussian. \;
\KwRet{$M=\tfrac 1r (RA')\T(RA')$ and the approximation $\tvec\beta^R = \arg\min_{\beta_d =-1} \vec\beta\T M \vec\beta$.}
\caption{Approximating Ridge Regression while Preserving Privacy\label{alg:private_JL}}
\end{algorithm}

\DeclareRobustCommand{\algJLEstimatew}{
\begin{algorithm}[htb]
\DontPrintSemicolon
\KwIn{A matrix $A \in \mathbb{R}^{n\times d}$ and a bound $B>0$ on the $l_2$-norm of any row in $A$.\\ \qquad Privacy parameters: $\epsilon, \delta > 0$.\\ \qquad Parameter $r_0$ indicating the minimal number of rows in the resulting matrix.}
Set $w = \sqrt{\tfrac {8B^2}\epsilon\left(\sqrt{2 r_0\ln(\tfrac 8 \delta)} + {\ln(\tfrac 8 \delta)}\right)}$. \;
Set $s \leftarrow \max\left\{0,\svmin(A\T A) - \frac{2B^2\ln(2/\delta)}\epsilon + Z\right\}$ where $Z\sim Lap(\frac{2B^2}\epsilon)$. \;
Adjust $w \leftarrow \sqrt{\max\{0, w^2 - s\} }$.\;
\eIf {$w>0$} {
Set $A'$ as the concatenation of $A$ with $wI_{d\times d}$.\;
Sample a $r\times (n+d)$-matrix $R$ whose entries are i.i.d samples from a normal Gaussian. \;
\KwRet{$M=\tfrac 1 {r_0} (RA')\T(RA')$, $w$ and the approximation $\tvec\beta^R = \arg\min_{\beta_d =-1} \vec\beta\T M \vec\beta$.}}
{
Set $r^* \leftarrow \max\left\{ r\in \mathbb{Z}:~\tfrac {8B^2}\epsilon\left(\sqrt{2 r\ln(\tfrac 8 \delta)} + {\ln(\tfrac 8 \delta)}\right)\leq s\right\}$.\;
Sample a $(r^*\times n)$-matrix $R$ whose entries are i.i.d samples from a normal Gaussian. \;
\KwRet{$M=\tfrac 1 {r^*} (RA)\T(RA)$, $r^*$ and the approximation $\tvec\beta = \arg\min_{\beta_d =-1} \vec\beta\T M \vec\beta$.}
}
\caption{Approximating Regression (Ridge or standard) while Preserving Privacy.\label{alg:private_JL_estimating_w}}
\end{algorithm}
}
\algJLEstimatew

To measure the effect of regularization we ran the following experiment. (Since the same experimental setting is used in the following sections we describe it here lengthly, and refer to it in later sections.) 

\subsection{The Basic Single-Regression Experiment --- Setting}
\label{subsec:experiment_setting}

To compare between the various algorithms we introduce and to analyze their utility we ran experiments testing their performance over data generated from a multivariate Gaussian. The experiments all share the same common setting, but each experiment studied a different set of estimators. In this section we detail the common setting, and in the next one we details the specific estimators and results of each experiment separately.

We pick $p=20$ i.i.d. features sampled from a normal Gaussian, and pick some $\vec\beta \in_R [-1,1]^{p+1}$ (the last coordinate denotes the regression's intercept), and set $\vec y$ as the linear combination of the features and the intercept (the all-$1$ column) plus random noise sampled from $\gauss{0,0.5}$. Hence our data had dimension $d=p+2 =22$ and the $21$-dimensional vector $\vec\beta$ has $l_2$ of about $3$. We vary $n$ to take any of the values in $\{2^{14}=4,096, 2^{15},2^{16},\ldots, 2^{25} = 33,554,432\}$. We vary $\epsilon$ to take any of the values $\{0.05,0.1, 0.15, 0.2, 0.25, 0.5\}$, and fix\footnote{We are aware that it is a good standard practice to set $\delta<\tfrac 1 n$ since otherwise, sampling from the data is $(\epsilon,\delta)$-differentially private. However, as we vary $n$ drastically, we aim to keep all other parameters equal.} $\delta=e^{-9}$, and use the $l_2$-bound of $B=\sqrt{2.5d}$. (As preprocessing, each datapoint whose length is $>B$ is shrunk to have length $B$.) For each estimator we experimented with, we run it $t=15$ times, and report the mean and standard variation of the $15$ experiments. In all experiments we measure the $l_2$-distance between the outputted estimator of each algorithm to the true $\vec\beta$ we used to generate the data. After all, the algorithms we give are aimed at learning the $\vec\beta$ that generated the given samples, and so they should return an estimator close to the true $\vec\beta$. We coded all experiments in \texttt{R} and ran the experiments on standard laptop.

\subsection{Experiment on Ridge Regression --- Measuring the Effect of Regularization}
\label{subsec:experiment_ridge}
To measure the effect of regularization we ran the following experiment in the setting detailed in Section~\ref{subsec:experiment_setting}. For each choice of $\epsilon$ and $n$ we ran three predictors. The first one is based on Algorithm~\ref{alg:private_JL_estimating_w} with $r_0=2d$. The second one is Algorithm~\ref{alg:private_JL} where we fed it the parameter $r$ that the first one used, just so all predictors will be comparable. The last one is the \emph{non-private} version that projected the data itself, without appending it with the $w\cdot I_{d\times d}$ matrix (again, using the same parameter $r$ as the other two predictors).  The results are given in Figure~\ref{fig:JL_comparison}.

\DeclareRobustCommand{\figJLcomparison}{
\begin{figure}[htb]
\begin{center}
\mypic{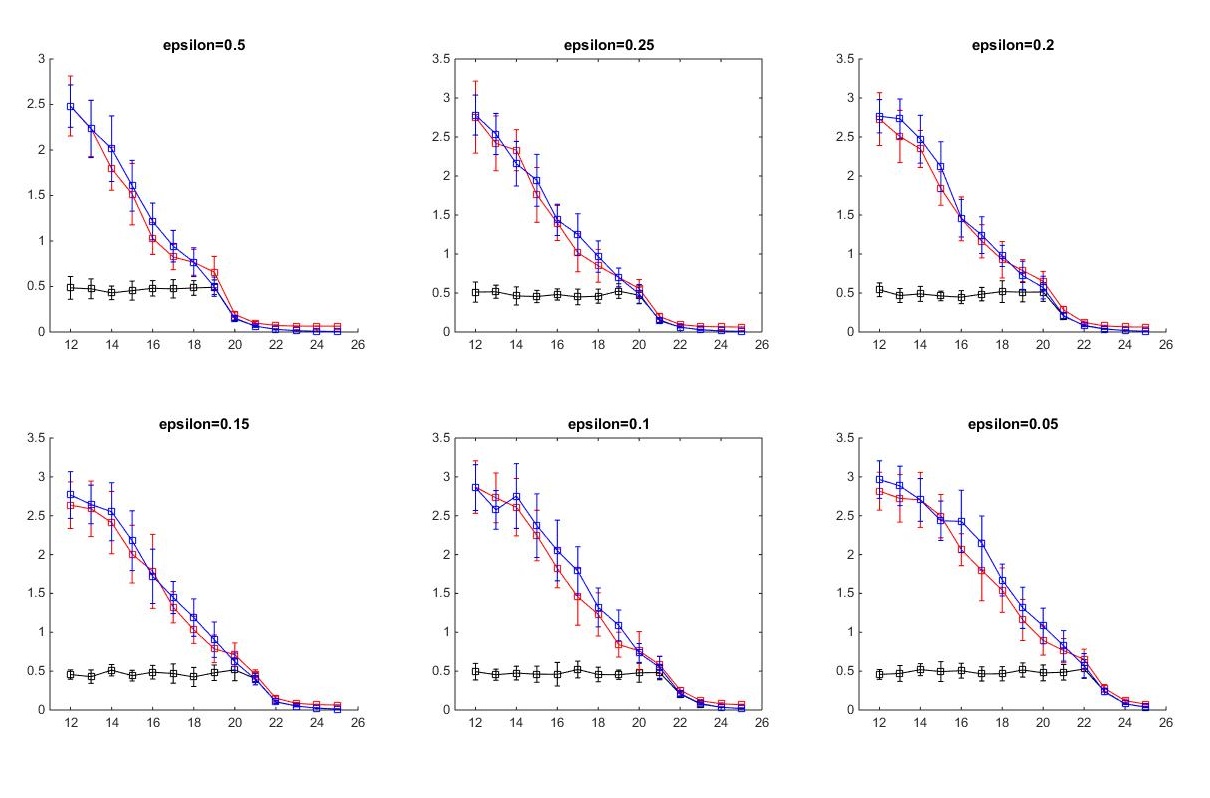} 
\caption{\label{fig:JL_comparison} (best seen in color) A comparison of the average $l_2$-error of the JL-based estimators. Algorithm~\ref{alg:private_JL_estimating_w} in blue, Algorithm~\ref{alg:private_JL} in red, and the non-private version in black. The $x$-axis is the size of the data in log-scale.}
\end{center}
\end{figure}
}
\figJLcomparison

The results are strikingly similar across all values of $\epsilon$. Initially the error of the predictors is very high (for the value $\vec\beta$ we used to generate the data, $\|\vec\beta\|\approx 2.786$, so such levels on noise mean in fact zero utility). Furthermore, it takes a while until Algorithm~\ref{alg:private_JL_estimating_w} (in blue) outperforms the more na\"ive Algorithm~\ref{alg:private_JL} (in red). (In most experiments, it happens only once $n\geq 2^{19}$ or $n\geq 2^{20}$.) This implies that the privacy-budget ``wasted'' on the private estimation of the least singular value of the data actually ends up reducing our utility but not by a large factor. Towards the largest value of $n$, Algorithm~\ref{alg:private_JL_estimating_w} actually does noticeably better than Algorithm~\ref{alg:private_JL} by a multiplicative factor of $\approx 3.5$ to $\approx 11$ (for $n=2^{25}$ when $\epsilon=0.1$ we have mean accuracy of $0.0192$ vs. $0.0671$; when $\epsilon=0.5$ we have mean accuracy of $0.0058$ vs. $0.0639$). In all experiments, the non-private estimator (in black) was clearly the best for all values of $n$.

\myvspace{-25pt}
\section{Additive Wishart Noise --- Regression with Additional Random Examples}
\label{sec:additive_Wishart}
\myvspace{-8pt}
As discussed in the previous section, Ridge regression can be viewed as regression where in addition to the sample points given by $[X;\vec y]$ we see $d$ additional datapoints given by $wI_{d\times d}$. Our second techniques follows this approach, only, instead of introducing these $d$ fixed datapoints, we introduce a few more than $d$ datapoints which are \emph{random} and independent of the data.\footnote{Independent of the data itself, but dependent of its properties. Our noise \emph{does} depend on the $l_2$-bound $B$.} Formally, we give the details in Algorithm~\ref{alg:additive_Wishart}  and immediately following --- the theorem proving it is $(\epsilon,\delta)$-differentially private.

\begin{algorithm}[htb]
\DontPrintSemicolon
\KwIn{A matrix $A \in \mathbb{R}^{n\times d}$ and a bound $B>0$ on the $l_2$-norm of any row in $A$.\\ \qquad Privacy parameters: $\epsilon, \delta > 0$.}
Set $k \leftarrow \lfloor d+ \tfrac{14}{\epsilon^2}\cdot 2\ln(4/\delta)\rfloor$.\;
Sample $\vec v_1, \vec v_2,\ldots, \vec v_k$ i.i.d examples from $\gauss{\vec 0_d, B^2 I_{d\times d}}$.\;
\KwRet{$M=A\T A + \sum_{i=1}^k \vec v_i \vec v_i\T$ and the approximation $\tvec\beta = \arg\min_{\beta_d =-1} \vec\beta\T M \vec\beta$.}
\caption{Additive Wishart Noise Algorithm\label{alg:additive_Wishart}}
\end{algorithm}
\begin{theorem}
\label{thm:Wishart_body}
Fix $\epsilon\in (0, 1)$ and $\delta \in(0,\tfrac 1 e)$. Fix $B>0$.
Let $A$ be a $(n\times d)$-matrix where each row of $A$ has bounded $l_2$-norm of $B$. Let $N$ be a matrix sampled from the $d$-dimensional Wishart distribution with $k$-degrees of freedom using the scale matrix $B^2\cdot I_{d\times d}$ (i.e., $N\sim \W_d(B^2\cdot I_{d\times d},k)$) for $k \geq \lfloor d+ \tfrac{14}{\epsilon^2}\cdot 2\ln(4/\delta)\rfloor$. 
Then outputting $X = A\T A + N$ is $(\epsilon,\delta)$-differentially private.
\end{theorem}

Note: Ridge Regression also has a Bayesian interpretation, as introducing a prior on $\vec\beta$ in regression problem. It is therefore tempting to argue that Theorem~\ref{thm:Wishart_body} implies that solving the regression problem with a random prior preserves privacy. (I.e., output the MAP of $\beta$ after setting its prior to a random sample from the Wishart distribution.) However, this analogy isn't fully accurate, since our algorithm also adds random noise to $X\T \vec y$. Indeed, regardless of what prior we use for $\vec\beta$, if $\vec y = \vec 0_n$ then we always output $\vec 0_p$ as the estimator of $\vec\beta$, so one can differentiate between the case that $\vec y = \vec 0_n$ and $\vec y \neq \vec 0_n$. We leave the (very interesting) question of whether Wishart additive random noise can be interpreted as a Bayesian prior for future work. 

We give a bound on the utility of the estimator we get with this technique in Theorem~\ref{thm:utility_additiveWishart}. However, we are more interested in the utility of this approach after we \emph{remove} some of the noise we add in this technique. Note, $\E[N] = kB^2 \cdot I_{d\times d}$, and so it stands to reason that we output $A\T A + N - kB^2\cdot I_{d\times d}$. Now, when $\svmin(A\T A)$ is small, we run the risk that some of the eigenvalues of $A\T A +N$ are smaller then $kB^2$, causing some of the eigenvalue of $A\T A + N - kB^2\cdot I_{d\times d}$ to be negative (which means we no longer output a PSD). In such a case, Lemma~\ref{lem:eigenvalues_Wishart} assures us that w.h.p we \emph{can} decrease $A\T A + N$ by $B^2\left(\sqrt{k}-(\sqrt{d}+\sqrt{2\ln(4/\delta)})\right)^2 \cdot I_{d \times d}$ and maintain the property that the output is positive definite matrix. This is the algorithm we set to evaluate empirically. 
\cut{
\myparagraph{Empirical evaluation of the utility of additive Wishart noise.} To evaluate the utility of the additive random Wishart noise algorithm we use the same setting to the one we used in Section~\ref{sec:ridge_regression}. Again, we pick $p=20$ i.i.d. independent features sampled from a normal Gaussian, and one linear combination  of the features (and an intercept) plus random noise sampled from $\gauss{0,0.5}$. Again, $n$ to takes any value in $\{2^{14}, 2^{15},2^{16},\ldots, 2^{25}\}$, and $\epsilon$ -- any value in $\{0.05,0.1, 0.15, 0.2, 0.25, 0.5\}$. For each choice of $\epsilon$ and $n$ we ran three predictors. The first one is the na\"ive and non-private linear regression, that uses the data with no additive noise (i.e., $\hvec\beta$). The second one is given by Algorithm~\ref{alg:additive_Wishart}. The last one is the estimator we get using the output of Algorithm~\ref{alg:additive_Wishart} minus either $kB^2 \cdot I_{d\times d}$ or $B^2\left(\sqrt{k}-(\sqrt{d}+\sqrt{2\ln(4/\delta)})\right)^2 \cdot I_{d \times d}$ (whichever of the two we can use and maintain positive definiteness). We repeat each experiment $t=15$ times, measuring the $l_2$-distance between the outputted estimator of each algorithm to the true $\vec\beta$ we used to generate the data. (This yields randomness in $\|\hvec\beta-\beta\|$, since every time we re-sample the data.) We report the mean and standard variation of the $15$ experiments. The results are given in Figure~\ref{fig:Wishart_comparison} in the supplementary material. The results are again consistent across the board --- reducing the noise also reduces the error, and indeed the second estimator is consistently doing better than the na\"ive estimator.
}
\DeclareRobustCommand{\figWishartcomparison}{
\begin{figure}[htb]
\begin{center}
\mypic{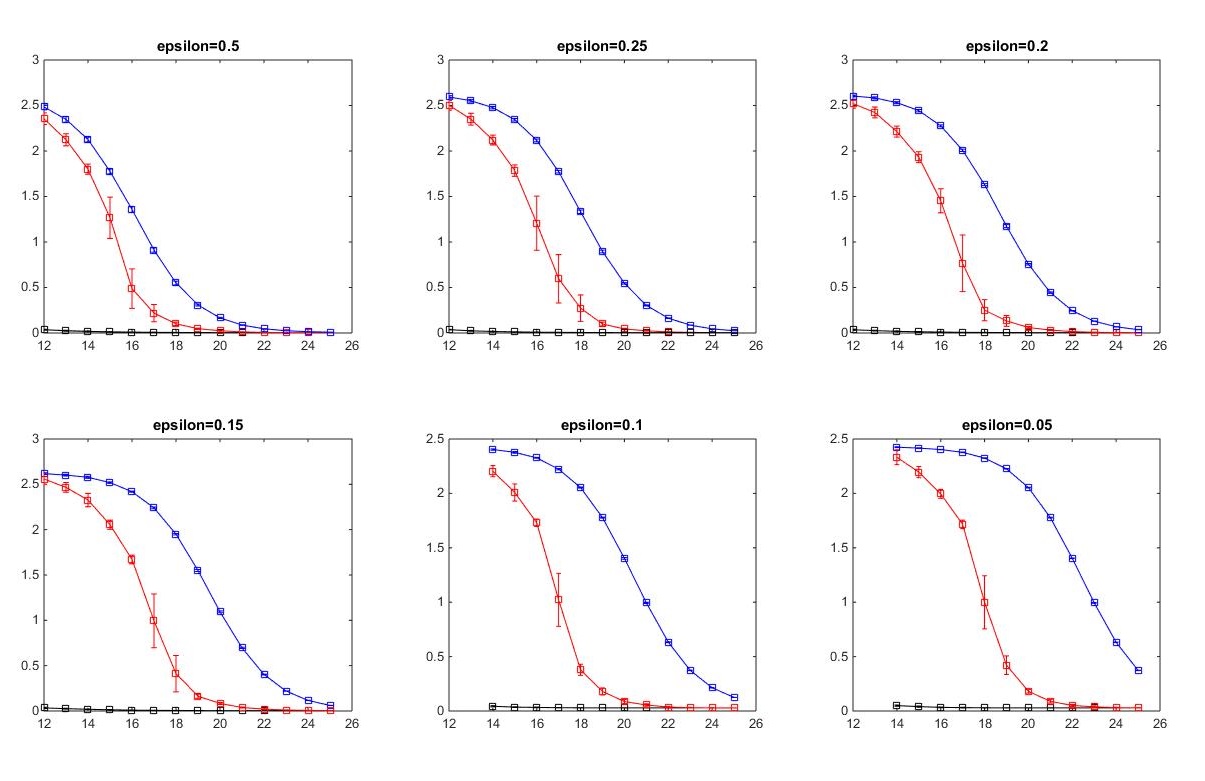} 
\caption{\label{fig:Wishart_comparison} (best seen in color) A comparison of the average $l_2$-error of the Wishart additive noise estimators. Algorithm~\ref{alg:additive_Wishart} in blue, deducting the expected shift from the output of Algorithm~\ref{alg:additive_Wishart} and then running the regression is in red, and the non-private version in black. The $x$-axis is the size of the data in log-scale.}
\end{center}
\end{figure}
}

\subsection{Experiment on Additive Wishart Noise}
\label{subsec:experiment_additive_wishart}

To evaluate the utility of the additive random Wishart noise algorithm we implemented and ran the algorithm in the same setting as detailed in Section~\ref{subsec:experiment_setting}. For each choice of $\epsilon$ and $n$ we ran three predictors. The first one is the na\"ive and non-private linear regression, that uses the data with no additive noise (i.e., $\hvec\beta$). The second one is given by Algorithm~\ref{alg:additive_Wishart}. The last one is the estimator we get using the output of Algorithm~\ref{alg:additive_Wishart} minus either $kB^2 \cdot I_{d\times d}$ or $B^2\left(\sqrt{k}-(\sqrt{d}+\sqrt{2\ln(4/\delta)})\right)^2 \cdot I_{d \times d}$ (whichever of the two we can use and maintain positive definiteness). We repeat each experiment $t=15$ times, measuring the $l_2$-distance between the outputted estimator of each algorithm to the true $\vec\beta$ we used to generate the data. (This yields randomness in $\|\hvec\beta-\vec\beta\|$, since every time we re-sample the data.) We report the mean and standard variation of the $15$ experiments. The results are given in Figure~\ref{fig:Wishart_comparison}. The results are again consistent across the board --- reducing the noise also reduces the error, and indeed the second estimator is consistently doing better than the na\"ive estimator.

\figWishartcomparison

\myvspace{-12pt}
\section{Sampling from an Inverse-Wishart Distribution (Bayesian Posterior)}
\label{sec:inv_wishart}
\myvspace{-8pt}
In Bayesian statistics, one estimates the 2nd-moment matrix in question by starting with a prior and updating it based on the examples in the data. More specifically, our dataset $A$ contains $n$ datapoints which we assumed to be drawn i.i.d from some $\gauss{\vec 0_d, V}$. We assume $V$ was sampled from some distribution $\mathcal{D}$ over positive definite matrices, which is the prior for $V$. We then update our belief over $V$ using the Bayesian formula: $\Pr[V \,|\, A] = \frac{ \Pr[A \,|\, V]\cdot \Pr_{\mathcal{D}}[V]}{ \int_W\Pr[A \,|\, W]\cdot \Pr_{\mathcal{D}}[W]dW}$. Finally, with the posterior belief we give an estimation of $V$ --- either by outputting the posterior distribution itself, or by outputting the most-likely $V$ according to the posterior, or by sampling from this posterior distribution (maybe multiple times). In this work we assume that our estimator of $V$ is given by sampling from the posterior distribution. 

One of the most common priors used for positive definite matrices is the inverse-Wishart distribution. This is mainly due to the fact that the inverse-Wishart distribution is conjugate prior.\footnote{A family of distributions is called conjugate prior if the prior distribution and the posterior distribution both belong to this family.} Specifically, if our prior belief is that $V \sim \W^{-1}_d(\Psi, k)$, then after viewing $n$ examples in the dataset $A$ our posterior is $V \sim \W^{-1}_d\left( (A\T A + \Psi), n+k\right)$.
Here we show that sampling such a positive definite matrix $V$ from our posterior inverse-Wishart distribution is $(\epsilon,\delta)$-differentially private, provided the prior distribution's scale matrix, $\Psi$, has a sufficiently large $\svmin(\Psi)$. This result is in line with the recent  beautiful work of Vadhan and Zheng~\cite{VadhanZ15}, who showed that many Bayesian techniques for estimating the means are differentially private, provided the prior is set correctly. The formal description of our algorithm and its privacy statement are given below.
\myvspace{-10pt}
\begin{algorithm}[htb]
\DontPrintSemicolon
\KwIn{A matrix $A \in \mathbb{R}^{n\times d}$ and a bound $B>0$ on the $l_2$-norm of any row in $A$.\\ \qquad Privacy parameters: $\epsilon, \delta > 0$.}
Set $\psi \leftarrow \frac{2B^2}\epsilon \left(2\sqrt{2(n+d)\ln(4/\delta)} + 2\ln(4/\delta) \right)$.\;
Sample $M \sim \W^{-1}_d((A\T A+ \psi\cdot I_{d\times d}), n+d)$.\;
\KwRet{$M$ and the approximation $\tvec\beta = \arg\min_{\beta_d =-1} \vec\beta\T M \vec\beta$.} 
\caption{Sampling from an Inverse-Wishart Distribution\label{alg:inv_Wishart}}
\end{algorithm}
\myvspace{-10pt}
\begin{theorem}
\label{thm:inverse_Wishart_body}
Fix $\epsilon>0$ and $\delta \in(0,\tfrac 1 e)$. Fix $B>0$. 
Let $A$ be a $(n\times d)$-matrix and fix an integer $\nu \geq d$. 
Let $w$ be such that $w^2={{2B^2 \left(2\sqrt{2\nu\ln(4/\delta)} + 2\ln(4/\delta) \right)}/{\epsilon}}$.
Then, given that $\sigma_{\min}(A) \geq w$, the algorithm that samples a matrix from $\W_d^{-1}(A\T A,\nu)$ is $(\epsilon,\delta)$-differentially private.
\end{theorem}

We comment on the similarities between Theorem~\ref{thm:inverse_Wishart_body} and Theorem~\ref{thm:JL_body}. Indeed, the Algorithm~\ref{alg:private_JL} essentially samples a matrix from $\W(A\T A + w^2 I, k)$ for some choice of $w$ and $k$ (and then normalizes the sample by $\tfrac 1 k$); and Algorithm~\ref{alg:inv_Wishart} samples a matrix from $\W^{-1}(A\T A + w^2 I, k)$ for a very similar choice of $w$. 
In fact, in Algorithm~\ref{alg:private_JL}, instead of sampling $R$ and then multiplying it with $A$, we can sample the same $R$ and multiply it with $(A\T A)^{1/2}$; or even sample a $(r\times d)$-matrix $\tilde R$ where each of its rows is sampled i.i.d from $\gauss{\vec 0_d, A\T A}$. (All of those have the same distribution over the output.) 
And so, much like we did in the Johnson-Lindenstrauss case, we can also use part of the privacy budget to estimate $\svmin(A\T A)$ and then set the parameter $\psi$ accordingly. Details appear in Algorithm~\ref{alg:inv_Wishart_estimating_w}. 
\DeclareRobustCommand{\alginvWishartEstimatew}{
\begin{algorithm}[htb]
\DontPrintSemicolon
\KwIn{A matrix $A \in \mathbb{R}^{n\times d}$ and a bound $B>0$ on the $l_2$-norm of any row in $A$.\\ \qquad Privacy parameters: $\epsilon, \delta > 0$.\\ \qquad A parameter $k_0$ indicating the minimal degrees of freedom.}
Set $\psi \leftarrow \frac{4B^2}\epsilon \left(2\sqrt{2k_0\ln(8/\delta)} + 2\ln(8/\delta) \right)$.\;
Set $s \leftarrow \max\left\{0,\svmin(A\T A) - \frac{2B^2\ln(2/\delta)}\epsilon + Z\right\}$ where $Z\sim Lap(\frac{2B^2}\epsilon)$. \;
Adjust $\psi \leftarrow {\max\{0, \psi - s\} }$.\;
\eIf {$w>0$} {
Sample $M \sim \W^{-1}_d((A\T A+ \psi\cdot I_{d\times d}), k_0)$.\;}
{
Set $k^* \leftarrow \max\left\{k\in \mathbb{Z} :\, \frac{4B^2}\epsilon \left(2\sqrt{2k\ln(8/\delta)} + 2\ln(8/\delta) \right)\leq s \right\}$\;
Sample $M \sim \W^{-1}_d(A\T A, k^*)$.\;
}
\KwRet{$M$ and the approximation $\tvec\beta = \arg\min_{\beta_d =-1} \vec\beta\T M \vec\beta$.}
\caption{Sampling from an Inverse-Wishart Distribution whose degrees of freedom are determined by the input.\label{alg:inv_Wishart_estimating_w}}
\end{algorithm}
}
\alginvWishartEstimatew
\cut{
\myparagraph{Measuring the utility of Algorithms~\ref{alg:inv_Wishart} and~\ref{alg:inv_Wishart_estimating_w}.} To estimate the utility of our algorithms, we conduct similar experiments to before. Again, we pick $p=20$ i.i.d. independent features sampled from a normal Gaussian, a single linear combination $\vec y$ of the features, an intercept and a random noise sampled from $\gauss{0,1}$. We vary $n$  and $\epsilon$ as before, and fix $\delta=e^{-9}$, and the $l_2$-bound of $B=\sqrt{2.5d}$. For each choice of $\epsilon$ and $n$ we ran 5 predictors. (i) The first one (in black) is the na\"ive and non-private Bayesian posterior sampling from the inverse Wishart distrbituion. (ii) The second one is given by Algorithm~\ref{alg:inv_Wishart} (in blue). (iii) The third one is given by Algorithm~\ref{alg:inv_Wishart_estimating_w} (in red). (iv) The fourth is given by Algorithm~\ref{alg:private_JL_estimating_w} (in green), and (v) the fifth one is the estimator we get using Algorithm~\ref{alg:private_JL_estimating_w} only using the inverse sampling, i.e., where we output $\left( (A\T A)^{-1/2} R\T R (A\T A)^{-1/2} \right)^{-1}$ (in magenta). We repeat each experiment $t=15$ times, measuring the $l_2$-distance between the outputted estimator of each algorithm to the true $\vec\beta$ we used to generate the data. We report the mean and standard variation of the $15$ experiments in Figure~\ref{fig:Wishart_Sampling_comparison}. The results are again consistent among the various choices of $\epsilon$. Both Algorithm~\ref{alg:inv_Wishart} and~\ref{alg:inv_Wishart_estimating_w} exhibit fairly large errors throughout, mainly due to the fact that the parameter $\psi$ used in each algorithm depends in $n$, as opposed to any other algorithm we present. We were surprised to see how little variance there exists in the results (the variance is too small to be visible in the figure). We did find it surprising that for the most part, the fact we split the privacy budget in Algorithm~\ref{alg:inv_Wishart_estimating_w} turns out to be consistently costlier than Algorithm~\ref{alg:inv_Wishart}, even for very large values of $n$. Another result that we found interesting is that technique (v) outperforms the JL-technique (iv) (and it is holds for all values of $n$). Initially we conjectured that the gap can be explained by Lemma~\ref{lem:JL_of_Wishart_and_invWishart}, where the bound for the inverse-JL has a slightly better second order term than the bound for the standard JL. However, for some values of $n$ the gap is fairly noticeable, and we leave it as an open problem to see if this holds for any projection matrix (and not just JL).
}
\DeclareRobustCommand{\figinvWishartcomparison}{
\begin{figure}[htb]
\begin{center}
\mypic{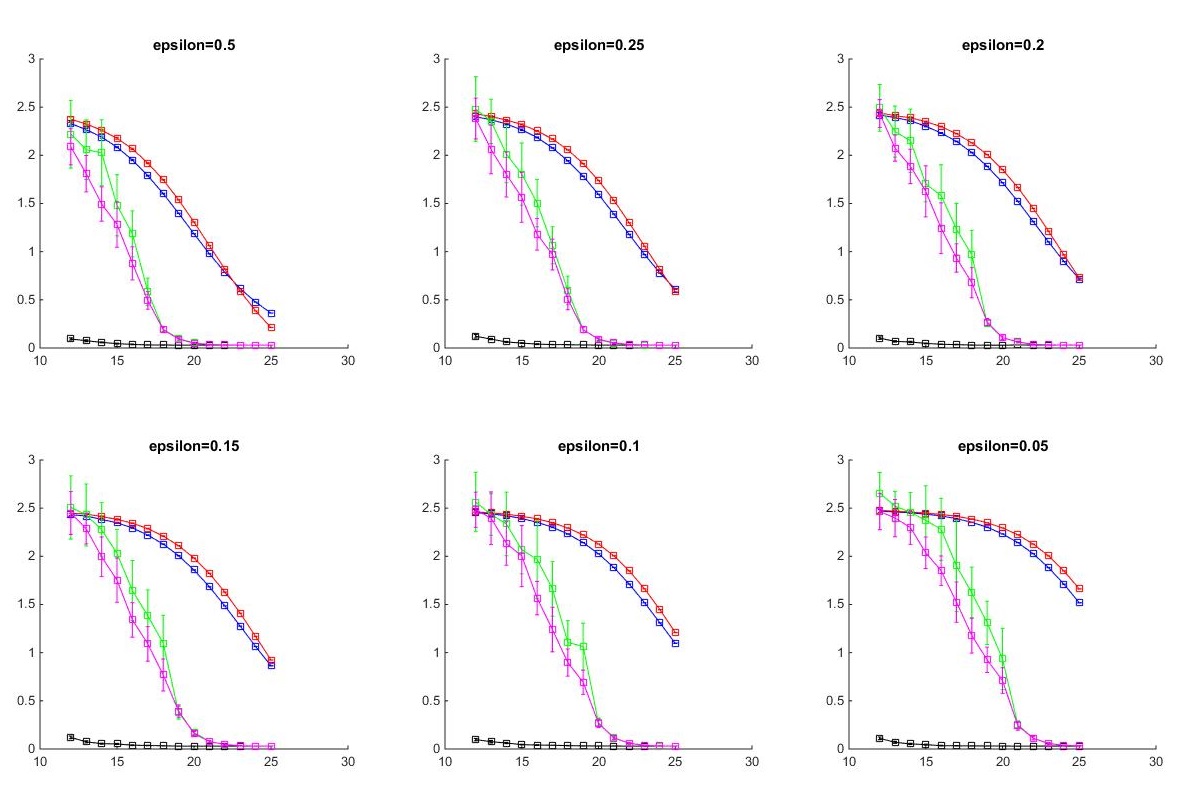} 
\caption{\label{fig:Wishart_Sampling_comparison} (best seen in color) A comparison of the average $l_2$-error for the estimators based on inverse-Wishart distribution sampling. The non-private sampler is in black, Algorithm~\ref{alg:inv_Wishart} is in blue and Algorithm~\ref{alg:inv_Wishart_estimating_w} in red. The JL-based algorithm (Algorithm~\ref{alg:private_JL_estimating_w}) that effectively samples from the Wishart distribution is in green; and its analogous algorithm that samples from the inverse-Wishart distribution (Algorithm~\ref{alg:inv_Wishart_estimating_w}) is in magenta. The $x$-axis is the size of the data in log-scale.}
\end{center}
\end{figure}
}
\subsection{Experiments on Sampling from the Inverse Wishart Distribution}
\label{subsec:experiment_invWishart}

To estimate the utility of Algorithms~\ref{alg:inv_Wishart} and~\ref{alg:inv_Wishart_estimating_w}, we conduct similar experiments to before, in the same setting detailed in Section~\ref{subsec:experiment_setting}. For each choice of $\epsilon$ and $n$ we ran 5 predictors. (i) The first one (in black) is the na\"ive and non-private Bayesian posterior sampling from the inverse Wishart distrbituion. (ii) The second one is given by Algorithm~\ref{alg:inv_Wishart} (in blue). (iii) The third one is given by Algorithm~\ref{alg:inv_Wishart_estimating_w} (in red) where the min-degrees-of-freedom parameter is set to $n+d$ (so that we have a direct way to compare between Algorithm~\ref{alg:inv_Wishart} and~\ref{alg:inv_Wishart_estimating_w}). (iv) The fourth is given by Algorithm~\ref{alg:private_JL_estimating_w} where the min-number-of-rows parameter is set $2d$ (in green), and (v) the fifth one is Algorithm~\ref{alg:inv_Wishart_estimating_w} when the min-degrees-of-freedom parameter is set $2d$. This gives us a direct comparison between Algorithms~\ref{alg:private_JL_estimating_w} and~\ref{alg:inv_Wishart_estimating_w}. We repeat each experiment $t=15$ times, measuring the $l_2$-distance between the outputted estimator of each algorithm to the true $\vec\beta$ we used to generate the data. We report the mean and standard variation of the $15$ experiments in Figure~\ref{fig:Wishart_Sampling_comparison}. The results are again consistent among the various choices of $\epsilon$. Both Algorithm~\ref{alg:inv_Wishart} and~\ref{alg:inv_Wishart_estimating_w} (techniques (ii) and~(iii)) exhibit fairly large errors throughout, mainly due to the fact that the parameter $\psi$ used in each algorithm depends in $n$, as opposed to any other algorithm we present. We were surprised to see how little variance there exists in the results (the variance is too small to be visible in the figure). We did find it surprising that for the most part, the fact we split the privacy budget in Algorithm~\ref{alg:inv_Wishart_estimating_w} turns out to be consistently costlier than Algorithm~\ref{alg:inv_Wishart}, even for very large values of $n$. Another result that we found interesting is that technique (v) outperforms the JL-technique (iv) (and it is holds for all values of $n$). Initially we conjectured that the gap can be explained by Lemma~\ref{lem:JL_of_Wishart_and_invWishart}, where the bound for the inverse-JL has a slightly better second order term than the bound for the standard JL. However, for some values of $n$ the gap is fairly noticeable, and we leave it as an open problem to see if this holds for any projection matrix (and not just JL).

\figinvWishartcomparison

\myvspace{-15pt}
\section{Comparison to the Analyze Gauss Baseline}
\label{sec:AG_comparison}
\myvspace{-8pt}
In this paper we discuss multiple ways for outputting a differentially private approximation of $A\T A$. One such way was already given by Dwork et al in their ``Analyze Gauss'' paper~\cite{DworkTTZ14}. As mentioned already, Dwork et al simply add to $A\T A$ a symmetric matrix $N$ whose entries are sampled i.i.d from a suitable Gaussian. Furthermore, the magnitude of the noise introduced by the Analyze Gauss algorithm is the smallest out of all algorithms. Yet, as we stressed before, the output of Analyze Gauss isn't necessarily a positive definite matrix. In this work we investigate the effect of these fact on the problem of linear regression. We study the utility of the Analyze Gauss algorithm for the linear regression problem both theoretically (the theorem regarding the utility of Analyze Gauss is deferred to Appendix~\ref{sec:utility_analysis}) and experimentally, in comparison to the other algorithms we introduce in this work. The high-level message from the experiments we show here as follows. In the simple case, Analyze Gauss is the best algorithm to use,,\footnote{In our opinion, this result is of interest by itself.} and when it returns ``unreasonable'' answers --- so do all other algorithms we use (details to follow). However, there do exist cases where it under performs in comparison to the additive Wishart noise algorithm (Algorithm~\ref{alg:additive_Wishart}) and the Wishart (Algorithm~\ref{alg:private_JL_estimating_w}) or inverse-Wishart (Algorithm~\ref{alg:inv_Wishart_estimating_w}) sampling algorithms.

In this section we compare between the following $6$ techniques.\\
\emph{1.} Analyze Gauss algorithm: output $A\T A + N$ with $N$ a symmetric matrix whose entries are i.i.d samples from a Gaussian (black line, squares.)\\
\emph{2.} The JL-based algorithm, Algorithm~\ref{alg:private_JL_estimating_w} (blue line, squares.)\\
\emph{3.} The additive Wishart noise algorithm given by Algorithm~\ref{alg:additive_Wishart} (magenta line, squares.)
\\
\emph{4.} A scaling version of Analyze Gauss: if the output of Analyze Gauss is not positive definite, add $cI_{d\times d}$ to it with $c = \E[\|N\|]$(black line, circles.)
\\
\emph{5.} Algorithm~\ref{alg:inv_Wishart_estimating_w}, which, as we commented in the experiments of Section~\ref{sec:inv_wishart}, is analogous to Algorithm~\ref{alg:private_JL_estimating_w} and seems to consistently do better than Algorithm~\ref{alg:private_JL_estimating_w}. Both Algorithm~\ref{alg:private_JL_estimating_w} and~\ref{alg:inv_Wishart_estimating_w} were given the same {min-degrees-of-freedom} parameter: $2d$ (blue line, circles.)
\\
\emph{6.} The scaling version of the additive Wishart random noise, as detailed in the experiment of Section~\ref{sec:additive_Wishart}. I.e., outputting $A\T A + W - k\cdot V$ (if this leaves the output positive definite) or $A\T A + W -(\sqrt k - (\sqrt d+\sqrt{2\ln(4/\delta)}))^2\cdot V$ otherwise (magenta line, circles.)

\myparagraph{Post-processing the Analyze Gauss output.} We have experimented extensively with multiple ways to project the output of the Analyze Gauss algorithm onto the manifold of PSD matrices. Indeed, the most na\"ive approach is to find a PSD matrix $M$ as to minimize $\|M - (A\T A + N)\|_F^2$. Such $M$ effectively turns to be the result of zeroing out all negative eigenvalue of $(A\T A+N)$. The utility of this approach turns out to be just as bad as the standard Analyze Gauss algorithm (with no post-processing), returning estimations of size $12$ or $9$ when the true $\vec\beta$ has $\|\vec\beta\|\approx 3$. Other approached we have experimented with were to try other values of $c$ for a post-processing of the form $A\T A + N + c I_{d\times d}$. (Such as setting $c$ to be the upper- and lower-bound on the singular values of $N$ w.p. $\geq 1-\delta$.) The performance of such approaches was, overall, comparable to the chosen technique (setting $c = \E[\|N\|]$) but with worse performance then our choice of $c$. Therefore, in our experiment, we used the \emph{best} of all techniques \emph{we were able to come up} with to post-process Analyze Gauss. This, however, does not mean that there isn't another post-process technique for Analyze Gauss that we didn't think of which out-performs our own approach.

\subsection{The Basic Single-Regression Experiment}
\label{subsec:experiment_ag_baseline} 

In the same experiment setting from before (see Section~\ref{subsec:experiment_setting}) we compare our 6 estimators based on the $l_2$-distance to the true $\vec\beta$ that generates our observations. The results, given in Figure~\ref{fig:comparison_single_regression} are pretty conclusive: Analyze Gauss is the better of all algorithms. Indeed, for smaller values of $n$ its output is completely out of scale (while $\|\vec\beta\|\approx 3$, the average error of Analyze Gauss is about $9$, $12$, and sometimes $30$). In fact, the error of Analyze Gauss for small values of $n$ is so large that we don't even present it in our graphs (and the standard deviation is so large, that the error bar of Analyze Gauss results in a big spike for such values of $n$). However, it is important to notice that for such values of $n$ \emph{all} other techniques also have a fairly large error (recall, $\|\vec\beta\|$ is roughly $3$, so errors $>2.5$ essentially give no information about $\vec\beta$). Once $n$ reaches a certain size, then there is a sharp shift transition, and Analyze Gauss becomes the algorithm with the smallest error for all greater values of $n$. Eventually, the errors of all algorithms becomes smaller than the error between $\hvec\beta$ and $\vec\beta$ ($\hvec\beta$ is the non-private estimator of $\vec\beta$). 
We also comment that, like before, technique \emph{5} (Algorithm~\ref{alg:inv_Wishart_estimating_w}) in consistently better than technique \#2 (Algorithm~\ref{alg:private_JL_estimating_w}), but also note that both technique have the largest variances in comparison to all other techniques.

\DeclareRobustCommand{\figSingleRegressionComparison}{
\begin{figure}[htb]
\begin{center}
\mypic{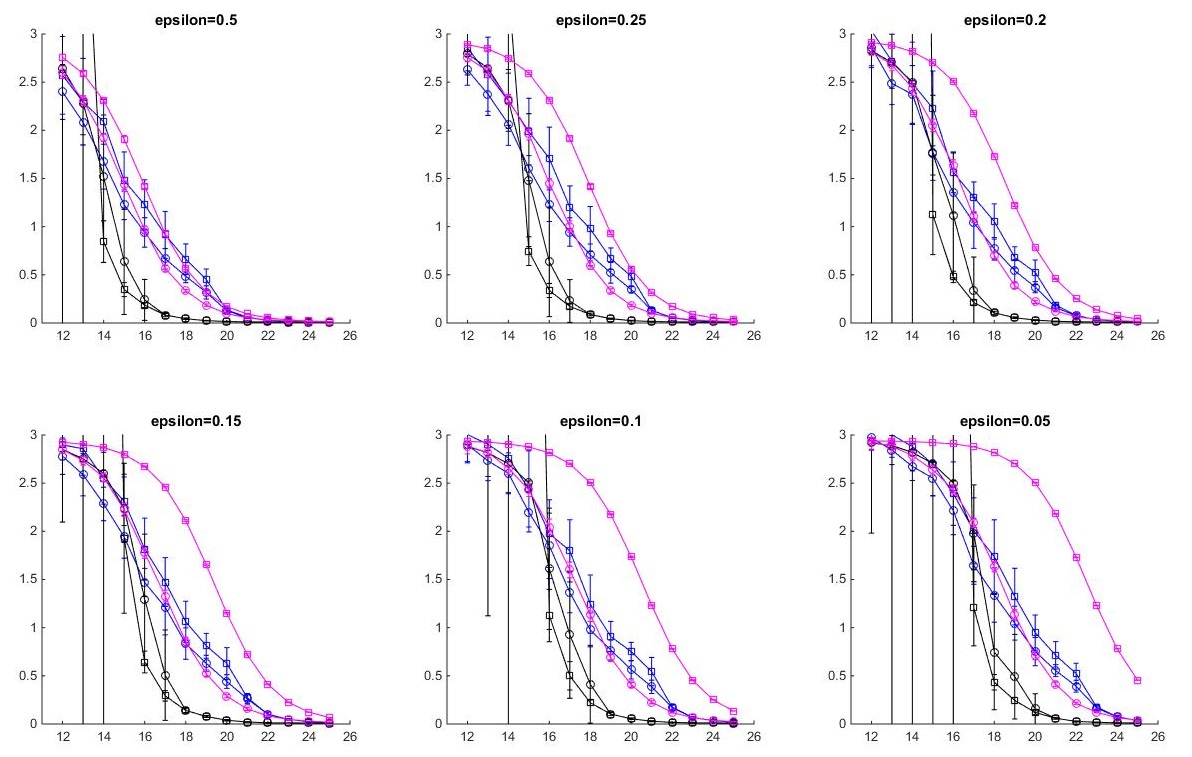} 
\caption{\label{fig:comparison_single_regression} (best seen in color) A comparison of the average $l_2$-error for our $6$ estimators. Analyze Gauss (squares) and its scaled version (circles) are in black; JL algorithm (squares) and the JL variant that samples from the inverse Wishart distribution (circles) are in blue; and the additive Wishart noise (squares) and its scaled version (circles) are in magenta. The $x$-axis is the size of the data in log-scale.}
\end{center}
\end{figure}
}
\figSingleRegressionComparison

\subsection{The Multiple-Regressions Experiment}
 
In this paper we argue that it is important to use algorithms that inherently output a positive definite matrix. To that end, we now investigate a more complex case, where the data is close to being singular, such that additive Gaussian noise is likely to introduce much error. The example we focus on is when the data $A$ is composed of $2p$ features: the first $p=20$ columns are independent of one another (sampled i.i.d from a normal Gaussian); the latter $p=20$ columns are the result of some linear combination of the first $p$ ones. And so $A = [X; \vec y_1, \ldots, \vec y_{p}]$ where for every $i$ we have $\vec y_i = X\vec\beta_i + \vec e_i$ where each coordinate of $\vec e_i$ is sampled i.i.d from $\gauss{0,\sigma^2}$ for $\sigma=0.5$ (fixed for all $i$). In our experiments, we vary $n$ (from $2^{12}$ to $2^{27}$ in powers of $2$), but fix $\epsilon=0.1$. What we also vary is the number of $\vec y$-features we use in our regression.

Recall, our algorithms approximate the Gram matrix of the data. Once such an approximation is published, it is possible to run as many linear regressions on it as we want --- fixing any one column of the data as a label and any \emph{subset} of the remaining columns as the features of the problem. This is precisely what we analyze here. We look at the linear regression problem where the label is some $\vec y_{i_0}$, and the features of the problem are the first $d$ columns plus some $m$ additional $\vec y$-columns.\footnote{We actually used one more column, of all $1$s, representing the intercept.} (I.e.: $\{\vec x_1, \ldots \vec x_p \} \cup \{\vec y_1, \ldots, \vec y_m\}$ where the latter are disjoint to $\vec y_{i_0}$.) A good approximation of $\vec\beta$ should therefore return some $\tvec\beta$ which is $0$ (or roughly $0$) on the latter $m$ coordinates. This corresponds to what we believe to be a high-level task a data-analyst might want to perform: finding out which features are relevant and which are irrelevant for regression.

The results in this case are far less conclusive and are given in Figure~\ref{fig:comparison_many_regression_mean}. When $m=0$, we are back to the case of a single regression (with no redundant features), and here Analyze Gauss (black, squares) out-performs all other algorithms once $n$ is large enough (in our case, $n\geq 2^{16}$). Yet, it is enough to set $m=1$ to get very different results. When $m>0$ it is evident that Analyze Gauss really performs badly~--- in fact, in most cases its values were far beyond the range of a reasonable approximation for $\vec\beta$ (taking values like $26$ and $45$ where $\|\vec\beta\|\approx 3.2$). The scaled version of Analyze Gauss (black, circles) does perform significantly better, yet --- it is not the best out of all algorithms. In fact, it is consistently worse than the JL-based algorithms (blue, circles and squares) and from the scaled version of the additive Wishart noise (magenta, circles) for $n<2^{22} = 4,194,304$. Note that as $m$ increases, all algorithms' errors become fairly large. 
In addition, Figure~\ref{fig:comparison_many_regression_var} shows the variance of our estimators. It is clear that the scaled version of Analyze Gauss has the smallest variance.\footnote{There is a spike up for the largest value of $n=2^{27}$, which \emph{is} recurring throughout our experiments --- most estimators have reasonable performance, but some have a really large error. One possible explanation could be related to the fact that we shrink sample points to have $l_2$-norm of at most $B$.} However, the scaled additive Wishart noise algorithm (magenta, circle) seems to have a good variance as well, and, as discussed, does out-perform the scaled Analyze Gauss algorithm for a wide range of values of $n$. 
\myvspace{-10pt}
\begin{figure}[htb]
\begin{center}
\mypicinmaintext{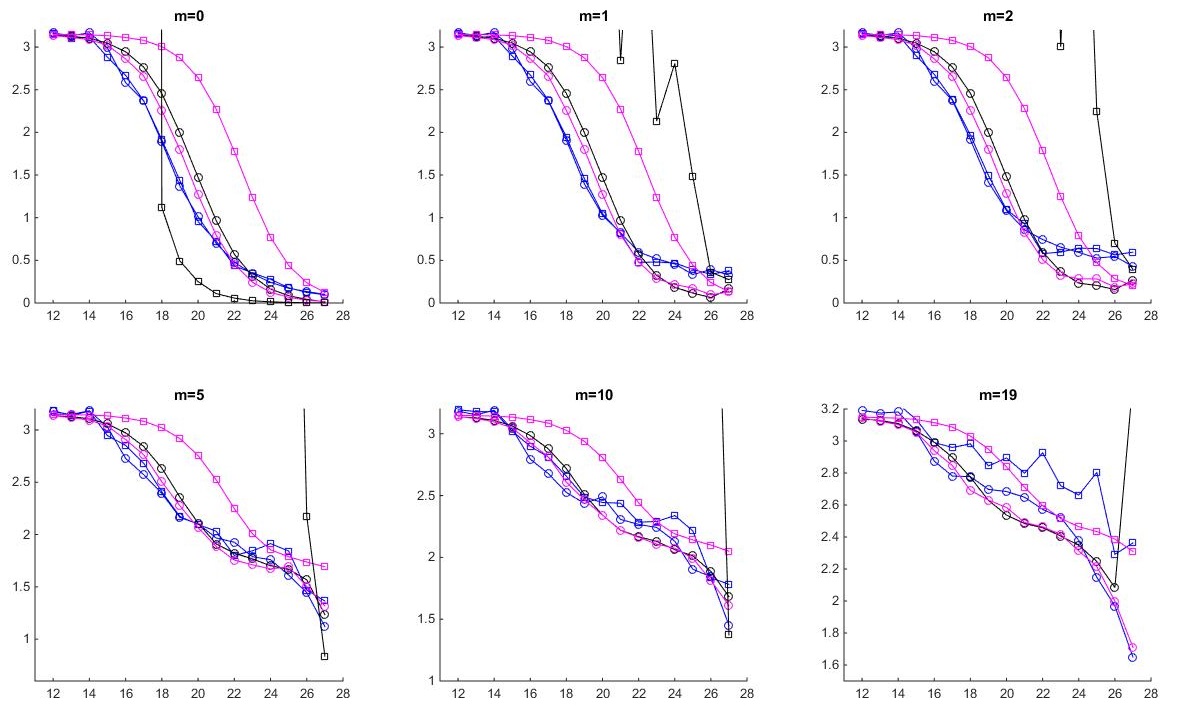} 
\caption{\label{fig:comparison_many_regression_mean} \small (best seen in color) A comparison of the average $l_2$-error for $6$ estimators based. Analyze Gauss (squares) and its scaled version (circles) are in black; JL algorithm (squares) and the JL variant that samples from the inverse Wishart distribution (circles) are in blue; and the additive Wishart noise (squares) and its scaled version (circles) are in magenta. The $x$-axis is the size of the data in log-scale.}
\end{center}
\end{figure}\myvspace{-10pt}

\DeclareRobustCommand{\figManyRegressionsVarComparison}{
\begin{figure}[htb]
\begin{center}
\mypic{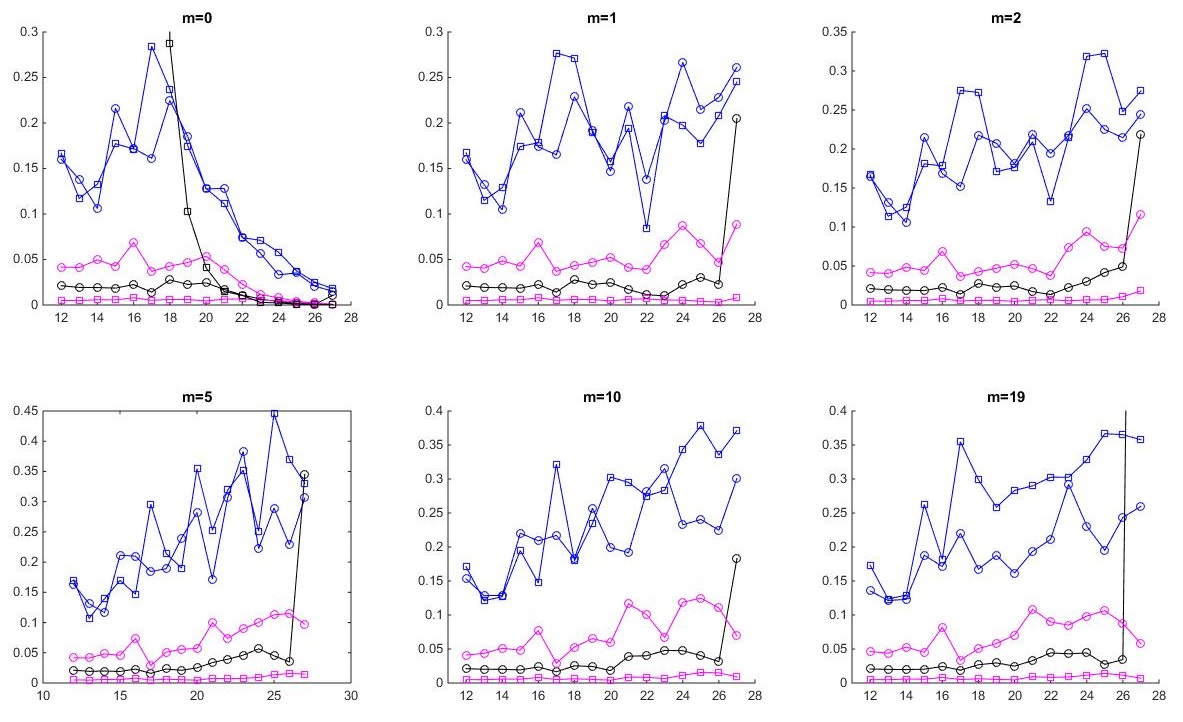} 
\caption{\label{fig:comparison_many_regression_var} (best seen in color) A comparison of the standard error of the $l_2$-error for $6$ estimators based, which are shown in Figure~\ref{fig:comparison_many_regression_mean}. Analyze Gauss (squares) and its scaled version (circles) are in black; JL algorithm (squares) and the JL variant that samples from the inverse Wishart distribution (circles) are in blue; and the additive Wishart noise (squares) and its scaled version (circles) are in magenta. The $x$-axis is the size of the data in log-scale.}
\end{center}
\end{figure}
}
\figManyRegressionsVarComparison
\myparagraph{Discussion.} It is possible to interpret the results of this experiment, especially for the larger values of $m$, as a detriment for \emph{all} the algorithms that approximate the Gram matrix of the data. Indeed, we pose the question of running regression over data where there does exist a large correlation between multiple columns as an open question. One approach could be to find a differentially private analogues to the techniques of~\cite{Mahoney11} for choosing a subset of the coordinates that approximate the $k$-PCA. An alternatively approach is to analyze the Lasso regression over the output of the algorithms that approximate the 2nd-moment matrix. In fact, we did experiment (though not extensively) with the Lasso regression. Using off-the-shelf Lasso regression packages (\texttt{R} package named \texttt{glmnet}), it seems that \emph{all} algorithms give estimators that are indeed sparse, but \emph{not} specifically over the latter $m$ coordinates. Rather, the estimator is sparse both on the the first $p$ coordinates and on the latter $m$ coordinates. In contract, running the Lasso regression on the data without additional randomness (non-privately) gives sparsity over the latter $m$ coordinates. We leave both problems for future work.

\bibliography{paper}

\appendix

\cut{
\section{Algorithms deferred from the main text}
\label{sec:deferred_algs}

For completeness, we specify here the formal description of algorithms deferred from the main text, as well as the figures corresponding to our various experiments.

\algJLEstimatew

\alginvWishartEstimatew
}

\section{Useful Lemmas}
\label{sec:supporting_lemmas}

In this section we detail the main lemmas that we use in our privacy proofs in the following section. The lemmas and theorems presented here, for the most part, were known prior to our work. We chose to include so that the uninformed reader can have their full proof, but we make no claim as to the originality of the proofs of the lemmas. The proofs of Lemma~\ref{lem:JL_of_Wishart_and_invWishart} and Claim~\ref{clm:chi_squared_tailbound} are based in part on the result Dasgupta and Gupta~\cite{DasguptaG03} and in part about results regarding the Wishart distribution given in~\cite{MardiaKB79} (Theorem 3.4.7). We encourage the reader who is familiar with lemmas and claims in this section to skip their proofs and turn to Section~\ref{sec:privacy_analysis} where we prove our privacy theorems.

\begin{lemma}
\label{lem:JL_of_Wishart_and_invWishart}
Let $X$ be a $(r\times d)$-matrix of i.i.d normal Gaussians (i.e., $x_{i,j}\sim\N(0,1)$). Fix $\beta \in (0,\tfrac 1 e)$. Then, for any vector $\vec v$ it holds that \[\Pr\left[\left| \vec v\T (\tfrac 1 r X\T X-I) \vec v\right| ~\leq~ \left(2\sqrt{\tfrac {2\ln(2/\beta)}{r}} + \tfrac {2\ln(2/\beta)}{r}\right) \|\vec v\|^2\right] \geq 1-\beta\]
Furthermore, if $r\geq d$ then denote $t = \sqrt{\frac {2\ln(2/\beta)}{r-d+1}}$ and assume $t<1$. Then
\begin{equation}\Pr\left[\left|\vec v\T (I- (\tfrac 1 {r-d+1} X\T X)^{-1}) \vec v\right| ~\leq~ \frac{2t-t^2}{(1-t)^2} \|\vec v\|^2\right] \geq 1-\beta\label{eq:newJLresult?}\end{equation}
\end{lemma}
\begin{proof}
Fix $\vec v$. Each entry of $X\vec v$ is distributed like $\N(0,\|v\|^2)$ and so $\vec v\T X\T X \vec v$ is just the sum of $r$ i.i.d Gaussians with variance $\|v\|^2$. In other words, $\tfrac 1 {\|\vec v\|^2}\vec v\T X\T X \vec v \sim \chi^2_r$. Concentration bounds (see Claim~\ref{clm:chi_squared_tailbound}) give therefore that w.p. $\geq 1-\beta$ we have
\[ (\sqrt r - \sqrt{2\ln(2/\beta)})^2  \leq  \tfrac 1 {\|\vec v\|^2}\vec v \T X\T X \vec v \leq (\sqrt r +\sqrt{2\ln(2/\beta)})^2 \] which implies \[\left(-2\sqrt{\tfrac {2\ln(2/\beta)}{r}} + \tfrac {2\ln(2/\beta)}{r}\right) \|\vec v\|^2 ~\leq~ \vec v\T (\tfrac 1 r X\T X-I) \vec v ~\leq~ \left(2\sqrt{\tfrac {2\ln(2/\beta)}{r}} + \tfrac {2\ln(2/\beta)}{r}\right) \|\vec v\|^2 \] and so we get the bound on $\vec v\T(\tfrac 1 r X\T X - I)\vec v$.

We now argue that $\frac {\vec v\T \vec v}{\vec v (X\T X)^{-1} \vec v} \sim \chi^2_{r-d+1}$. To see this, we argue that specifically for the vector $\vec e_d$ (the indicator of the $d$-th coordinate) we have $\frac 1 {\vec e_d (X\T X)^{-1} \vec e_d} \sim \chi^2_{r-d+1}$, and the results for any $\vec v$ follows from taking any unitary function s.t. $U\T \vec v = \|\vec v\| \vec e_d$, and the observation that the distributions of $X$ and $XU\T$ are identical.

Now, clearly $\vec e_d (X\T X)^{-1} \vec e_d = (X\T X)^{-1}_{d,d}$. Now, if we denote the last column of $X$ as $\vec x_d$ and the first $d-1$ columns of $X$ as $X_{-d}$ then $X\T X = \left[ \begin{array}{c c c | c}
& & & \cr
& X_{-d}\T X_{-d}& & X_{-d}\T \vec x_d\cr
& & &\cr
\hline& \vec x_d\T X_{-d} & &\|\vec x_d\|^2\cr
\end{array}\right]$. Thus, the formula for the entries of the inverse give
\begin{align*} \frac 1 {(X\T X)^{-1}_{d,d}} &= \|\vec x_d\|^2 - \vec x_d \T X_{-d} (X_{-d}\T X_{-d})^{-1} X_{-d}\T\vec x_d \cr
& = \vec x_d \left(I - X_{-d} (X_{-d}\T X_{-d})^{-1} X_{-d}\T \right) \vec x_d \stackrel{\rm def} = \vec x_d\T P~ \vec x_d
\end{align*}
Now, w.p. $1$ we have that $X_{-d}$ has full rank ($d-1$). For any choice of $X_{-d}$ with full rank we get a matrix $P$ which has rank $r-(d-1)$ and its eigenvalues are either $1$ or $0$. Hence, for any $X_{-d}$ we get $\tfrac 1 {(X\T X)^{-1}_{d,d}} \sim \chi^2_{r-d+1}$. Since this distribution is independent of $X_{-d}$ we therefore have that this result holds w.p. $1$. I.e.:
\begin{align*}
&\PDF\left(\tfrac 1 {(X\T X)^{-1}_{d,d}}=z\right) \cr
&~~= \int\displaylimits_P \PDF\left(\tfrac 1 {(X\T X)^{-1}_{d,d}}=z \,\big\vert\, I - X_{-d} (X_{-d}\T X_{-d})^{-1} X_{-d}\T=P\right) \cdot \PDF\left(I - X_{-d} (X_{-d}\T X_{-d})^{-1} X_{-d}\T=P \right)dP 
\cr &~  =\int\displaylimits_P \PDF_{\chi^2_{r-d+1}}(z) \cdot \PDF\left(I - X_{-d} (X_{-d}\T X_{-d})^{-1} X_{-d}\T=P \right)dP 
\cr &~= \PDF_{\chi^2_{r-d+1}}(z) \cdot \int\displaylimits_P \PDF\left(I - X_{-d} (X_{-d}\T X_{-d})^{-1} X_{-d}\T=P \right)dP = \PDF_{\chi^2_{r-d+1}}(z)
\end{align*}

Therefore, with probability $\geq 1-\beta$ we have \[\frac {\vec v\T \vec v}{\vec v\T (X\T X)^{-1} \vec v} \in \left( (\sqrt{r-d+1}-\sqrt{2\ln(2/\beta)})^2, (\sqrt{r-d+1}-\sqrt{2\ln(2/\beta)})^2\right)\] so 
\[ \left( \frac {\sqrt{r-d+1}}{\sqrt{r-d+1}+\sqrt{2\ln(2/\beta)}}\right)^2 \|\vec v\|^2\leq \vec v\T \left(\tfrac 1 {r-d+1}X\T X \right)^{-1} \vec v \leq \left( \frac {\sqrt{r-d+1}}{\sqrt{r-d+1}-\sqrt{2\ln(2/\beta)}}\right)^2 \|\vec v\|^2\] 
which implies
\begin{align*}
&\vec v\T \left( I - (\tfrac 1 {r-d+1}X\T X)^{-1} \right) \vec v \leq \frac {2\sqrt{\frac{2\ln(2/\beta)}{r-d-1}}+\frac{2\ln(2/\beta)}{r-d-1}} {(1+\sqrt{\frac{2\ln(2/\beta)}{r-d-1}} )^2}
\cr &\vec v\T \left( I - (\tfrac 1 {r-d+1}X\T X)^{-1} \right) \vec v \geq - \frac {2\sqrt{\frac{2\ln(2/\beta)}{r-d-1}}-\frac{2\ln(2/\beta)}{r-d-1}} {(1-\sqrt{\frac{2\ln(2/\beta)}{r-d-1}} )^2}
\end{align*}
Some arithmetic manipulations show that when $\frac{2\ln(2/\beta)}{r-d-1} < 1$ we have that 
\[ \left|\vec v\T \left( I - (\tfrac 1 {r-d+1}X\T X)^{-1} \right) \vec v\right| \leq \frac {2\sqrt{\frac{2\ln(2/\beta)}{r-d-1}}-\frac{2\ln(2/\beta)}{r-d-1}} {(1-\sqrt{\frac{2\ln(2/\beta)}{r-d-1}} )^2}\] as this is the larger term of the two.
\end{proof}

\begin{claim}
\label{clm:chi_squared_tailbound}
Fix $k$ and let $X_1,\ldots, X_k$ be iid samples from $\gauss{0,1}$.  Then, for any $0<\Delta <k$ we have that $\Pr[\sum_i X_i^2 > (\sqrt k + \sqrt\Delta)^2] < e^{-\Delta/2}$ and  $\Pr[\sum_i X_i^2 < (\sqrt k - \sqrt\Delta)^2] < e^{-\Delta/2}$.
\end{claim}
\begin{proof}
We start with the following calculation. For any $X\sim \gauss{0,1}$ and any $s < 1/2$ it holds that
\begin{align*}
\E[e^{sX^2}] & = \int\displaylimits_{-\infty}^\infty \tfrac 1 {\sqrt{2\pi}} e^{-\tfrac {x^2} 2} e^{sx^2} dx = \int\displaylimits_{-\infty}^\infty \tfrac 1 {\sqrt{2\pi}} e^{-\tfrac {x^2(1-2s) }2}dx 
&\stackrel{\substack{ y = x\sqrt{1-2s} \\ \textrm{ so }dy=dx\sqrt{1-2s} }} = \int\displaylimits_{-\infty}^{\infty} \tfrac 1 {\sqrt{2\pi}} e^{-\tfrac {y^2} 2} \tfrac {dy} {\sqrt{1-2s}} = \tfrac {1} {\sqrt{1-2s}}
\end{align*}
We now use Markov's inequality, to deduce that for any $\lambda \in (0,1/2)$
\begin{align*}
\Pr[\sum_i X_i^2 > (\sqrt k + \sqrt\Delta)^2] &= \Pr[e^{\lambda \sum_{i}X_i^2} > e^{\lambda(\sqrt k + \sqrt\Delta)^2}] \leq \frac {\E[e^{\lambda \sum_{i}X_i^2}]} {e^{\lambda(\sqrt k + \sqrt\Delta)^2}} = \prod_i \E[e^{\lambda X_i^2}]e^{-\lambda(\sqrt k + \sqrt\Delta)^2}
\cr &= \left( \frac 1 {1-2\lambda}\right)^{\tfrac k 2} {e^{-\lambda(\sqrt k + \sqrt\Delta)^2}} =\left( 1+\frac {2\lambda} {1-2\lambda}\right)^{\tfrac k 2} {e^{-\lambda(\sqrt k + \sqrt\Delta)^2}} 
\cr &\leq \exp\left(\frac{\lambda k}{1-2\lambda} -\lambda (\sqrt k + \sqrt\Delta)^2 \right)
\end{align*}
Setting $\lambda = \tfrac {\sqrt\Delta}{2(\sqrt k + \sqrt \Delta)}$ so that $1-2\lambda = \frac{\sqrt k}{\sqrt k +\sqrt \Delta}$ we have
\begin{align*}
\Pr[\sum_i X_i^2 > (\sqrt k + \sqrt\Delta)^2] &\leq \exp\left( \lambda \sqrt k (\sqrt k +\sqrt\Delta) - \lambda (\sqrt k + \sqrt \Delta)^2\right)
\cr &= \exp\left(\tfrac 1 2 \sqrt{k\Delta} - \tfrac 1 2 \sqrt\Delta(\sqrt k + \sqrt\Delta) \right) = \exp(-\tfrac \Delta 2)
\end{align*}

A similar calculation shows the lower bound.
\begin{align*}
\Pr[\sum_i X_i^2 < (\sqrt k - \sqrt\Delta)^2] &= \Pr[e^{-\lambda \sum_{i}X_i^2} > e^{-\lambda(\sqrt k - \sqrt\Delta)^2}] \leq \prod_i \E[e^{-\lambda X_i^2}]e^{\lambda(\sqrt k - \sqrt\Delta)^2}
\cr &= \left( \frac 1 {1+2\lambda}\right)^{\tfrac k 2} {e^{\lambda(\sqrt k - \sqrt\Delta)^2}} =\left( 1-\frac {2\lambda} {1+2\lambda}\right)^{\tfrac k 2} {e^{\lambda(\sqrt k - \sqrt\Delta)^2}} 
\cr &\leq \exp\left(-\frac{\lambda k}{1+2\lambda} +\lambda (\sqrt k - \sqrt\Delta)^2 \right)
\end{align*}
Setting $\lambda = \tfrac {\sqrt\Delta}{2(\sqrt k - \sqrt \Delta)}$ so that $1+2\lambda = \frac{\sqrt k}{\sqrt k -\sqrt \Delta}$ we have
\begin{align*}
\Pr[\sum_i X_i^2 > (\sqrt k + \sqrt\Delta)^2] &\leq \exp\left( -\lambda \sqrt k (\sqrt k -\sqrt\Delta) + \lambda (\sqrt k - \sqrt \Delta)^2\right)
\cr &= \exp\left(-\tfrac 1 2 \sqrt{k\Delta} + \tfrac 1 2 \sqrt\Delta(\sqrt k - \sqrt\Delta) \right) = \exp(-\tfrac \Delta 2)
\end{align*}
\end{proof}

\begin{lemma}
\label{lem:eigenvalues_Wishart}
Fix $\delta\in(0,e^{-1})$. Let $X$ be a matrix sampled from a Wishart distribution $\W_d(V,m)$ where $\sqrt m > \term$.  Then, w.p. $\geq 1-\delta$ we have that for every $j=1,2,\ldots, d$ it holds that
\[ \sigma_{j}(X) \in (\sqrt{m}\pm\term)^2 \sigma_{j}(V)\]
Furthermore, we also have that for any $0< \alpha \leq m$ it holds
\begin{align*}
&\left\| \alpha V - X \right\|\leq \|V\|\cdot |\alpha - (\sqrt m - \term)^2| ~~\textrm{ and }
\cr & \left\| (\alpha V)^{-1} - X^{-1} \right\|\leq \sigma_{\min}^{-1}(V)\cdot |\alpha^{-1}-(\sqrt m +\term)^{-2}|
\end{align*}
\end{lemma}
\begin{proof}
In order to sample $X \sim \W_d(V,m)$ we first sample a matrix $Y \in \mathbb{R}^{m\times d}$ in which every entry is i.i.d normal Gaussian. We then multiply $Y$ by $V^{1/2}$, s.t. every row in $Y V^{1/2}$ is sampled i.i.d from $\mathcal{N}(\vec 0_d, V)$. We then set $X = V^{1/2} Y\T Y V^{1/2}$.

Now, we invoke a theorem of Davidson and Szarek~\cite{DavidsonS01} (Theorem II.13) that states that for any $t>1$ we have
\[ \Pr[\sigma_{\max}(Y) > \sqrt{m}+\sqrt{d}+t] < e^{-t^2/2} ~~\textrm{ and }\Pr[\sigma_{\min}(Y) < \sqrt{m}-\sqrt{d}-t] < e^{-t^2/2}\] to deduce that w.p. $\geq 1 -\delta$ it holds that all of the singular values of $Y$ lie on the interval $\left( \sqrt{m}-\term, \sqrt{m}+\term\right)$. Next, we let $\vec u_j$ denote the $j$-th eigenvector of $V$, corresponding to the $j$-th eigenvalue $\sigma_j(V)$. Therefore, for any $j$ we have
\begin{align*}
\vec u_j \T X \vec u_j &= (V^{1/2}\vec u_j)^\T Y\T Y (V^{1/2}\vec u_j)\cr
& \leq (\sqrt m + \term)^2 \|V^{1/2}\vec u_j\|^2 = \sigma_j(V) (\sqrt m+\term)^2  \cr
\vec u_j \T X \vec u_j & \geq (\sqrt m - \term)^2 \|V^{1/2}\vec u_j\|^2 =\sigma_j(V) (\sqrt m-\term)^2
\end{align*}
and furthermore, for any subspace $S$ we have that
\begin{align*}
\max_{\vec u\in S:\, \|u\|=1} \vec u \T X \vec u 
& \leq (\sqrt m + \term)^2 \left(\max_{\vec u\in S:\, \|u\|=1} \|V^{1/2}\vec u_j\|^2\right)\cr
\min_{\vec u\in S:\, \|u\|=1} \vec u \T X \vec u 
& \geq (\sqrt m - \term)^2 \left(\min_{\vec u\in S:\, \|u\|=1} \|V^{1/2}\vec u_j\|^2\right)
\end{align*}

Thus, to complete the first part of the proof, we invoke the Courant-Fischer Min-Max Theorem that state that
\begin{align*}
\sigma_j(X) = \max_{\{S\subset \mathbb{R}^d:~ \dim(S) = j\}} \min_{\{\vec u\in S:~ \|\vec u\|=1\}} \vec u\T X \vec u = \min_{\{S\subset \mathbb{R}^d:~ \dim(S) = d-j+1\}} \max_{\{\vec u\in S:~ \|\vec u\|=1\}} \vec u\T X \vec u
\end{align*}
Therefore, we can pick $S' = span\{\vec u_1, \ldots \vec u_j\}$ and $S'' = span\{\vec u_j, \ldots, \vec u_d\}$ to deduce
\begin{align*}
\sigma_j(X) & \geq \min_{\vec u\in S': \|\vec u\|=1} \vec u\T X \vec u \geq (\sqrt m-\term)^2 \sigma_{j}(V)\cr
\sigma_j(X) & \leq \max_{\vec u\in S'': \|\vec u\|=1} \vec u\T X \vec u \leq (\sqrt m+\term)^2 \sigma_{j}(V)
\end{align*}
As for the second part of the claim, it follows from the fact that $\alpha V - X = V^{1/2} \left( \alpha I - Y\T Y \right) V^{1/2}$. Now, if we denote $ Y = U \Sigma U\T$ as the SVD decomposition of $Y$, we have $\alpha I - Y\T Y = U\left(  \alpha I - \Sigma \right) U\T$. Since all the entries on the diagonal lie in the range $|\alpha - (\sqrt m \pm \term)^2|$. As $\alpha \leq m$ we have that all eigenvalues are upper bounded by $(m-\alpha) + 2\sqrt{m}\term$ and the claim follows. Similarly, for $(\alpha V)^{-1} - X^{-1} = V^{-1/2}\left( \alpha I - Y\T Y \right) V^{-1/2}$ all eigenvalues lie in the range $|\alpha^{-1} - (\sqrt m \pm \term)^{-2}|$, which in this case is upper bounded by $|\alpha^{-1} - (\sqrt m + \term)^{-2}|$. We comment that the bounds on $\| \alpha V - X\|$ and on $\| (\alpha V)^{-1} - X^{-1}\|$ require we use both the upper- and lower-bounds on the eigenvalues of $Y$.

\end{proof}

The other two useful tools we use are the formula for rank-$1$ updates of the determinant and the inverse (the Sherman-Morrison lemma).
\begin{theorem}
\label{thm:Sherman_Morrison}
Let $A$ be a $(d\times d)$-invertible matrix and fix any two $d$-dimensional vectors $\vec u,\vec v$ s.t. $\vec v\T A^{-1} \vec u \neq -1$. Then:
\begin{align*}
&\det(A + \vec u \vec v\T) = \det(A)(1+\vec v\T A^{-1} \vec u)\cr
&(A+\vec u\vec v\T)^{-1} = A^{-1} - \frac {A^{-1} \vec u \vec v\T A^{-1}}{1+ \vec v\T A^{-1} \vec u}
\end{align*}
\begin{proof}
Since we have $A + \vec u \vec v\T = A (I + A^{-1} \vec u \vec v\T)$, we analyze the spectrum of the matrix $I + A^{-1} \vec u \vec v\T$. Clearly, for any $\vec x \perp \vec v$ we have $(I + A^{-1} \vec u \vec v\T)\vec x = \vec x + 0\cdot A^{-1}\vec u = \vec x$, so $d-1$ of the eigenvalues of $I + A^{-1} \vec u \vec v\T$ are exactly $1$. As for the last one, take a unit length vector $\vec z = \tfrac 1 {\|\vec v\|}\vec v$, and we have $\vec z\T (I + A^{-1} \vec u \vec v\T) \vec z = 1 + \|\vec v\| \cdot \vec z\T A^{-1} u =1 + \vec v\T A^{-1} \vec u$. Therefore, $\det(A + \vec u \vec v\T) = \det(A)\det(I + A^{-1} \vec u \vec v\T)= \det(A)(1+\vec v\T A^{-1} u)$.

As for the Sherman-Morrison formula, we can simply check and see that indeed:
\begin{align*}
(A+\vec u \vec v\T )(A^{-1} - \frac {A^{-1} \vec u \vec v\T A^{-1}}{1+ \vec v\T A^{-1} \vec u}) & = I + \vec u \vec v\T A^{-1} - \frac{\vec u \vec v\T A^{-1}}{1+ \vec v\T A^{-1} \vec u} - \frac{\vec u \vec v\T A^{-1}\vec u \vec v\T A^{-1}}{1+ \vec v\T A^{-1} \vec u}\cr
& = I + \vec u\vec v\T A^{-1}\left(1 - \frac 1 {1+\vec v\T A^{-1} \vec u} - \frac  {\vec v\T A^{-1} \vec u} {1+\vec v\T A^{-1} \vec u}\right) = I
\end{align*}
\end{proof}
\end{theorem}

\section{Privacy Theorems}
\label{sec:privacy_analysis}

In this section, we provide the formal proofs the our algorithms are differential privacy. We comment that, because we hope these algorithms will be implemented, we took the time to analyze the exact constants in our proofs rather than settling for $O(\cdot)$-notation. In addition to the three algorithms we provide, we give another theorem about the privacy of an algorithm that adds Gaussian noise to the inverse of the data, which may be of independent interest. 

\subsection{Privacy Proof for Algorithm~\ref{alg:private_JL}}

\begin{theorem}
\label{thm:JL}
Fix $\epsilon>0$ and $\delta \in(0,\tfrac 1 e)$. Fix $B>0$. 
Fix a positive integer $r$ and let $w$ be such that
\[ w^2 = B^2\left( 1+ \frac{ 1+ \tfrac{\epsilon}{\ln(4/\delta)}}\epsilon\left(2\sqrt{2 r\ln(\tfrac 4 \delta)} + {2\ln(\tfrac 4 \delta)}\right)\right)\]
Let $A$ be a $(n\times d)$-matrix with $d<r$ and where each row of $A$ has bounded $L_2$-norm of $B$. Given that $\sigma_{\min}(A) \geq w$, the algorithm that picks a $(r\times n)$-matrix $R$ whose entries are iid samples from a normal  distribution $\gauss{0,1}$ and publishes $R\cdot A$ is $(\epsilon,\delta)$-differentially private.
\end{theorem}
\begin{corollary}
assuming $\epsilon<1$ and $\delta < e^{-1}$, if it holds that $r \geq 2\ln(\tfrac 4 \delta)$ then it suffices to have $w^2 \geq 8B^2\frac{\sqrt{r\ln(4/\delta)}}{\epsilon}$ for the results of Theorem~\ref{thm:JL} to hold. Alternatively, given input where its least singular value is publicly known to $w$, we can set 
\[ r = \left \lceil \left( \frac {\epsilon w^2}{8B^2\ln(\tfrac 4\delta)}\right)^2\right \rceil, ~~\textrm{ if indeed }~ r>2\ln(\tfrac 4\delta)\]
and satisfy $(\epsilon,\delta)$-differential privacy. Therefore, if the rows of $A$ are i.i.d draws from a $\vec 0$-mean multivariate Gaussian with variance $\Sigma$, then we may set $r$ as $\left \lceil \left( n\frac {\epsilon \sigma_{\min}(\Sigma)}{8B^2\ln(\tfrac 4\delta)}\right)^2\right \rceil= \Omega(n^2)$.
\end{corollary}

\begin{proof}
Fix $A$ and $A'$ be two neighboring $(n\times d)$ matrices, s.t. $A-A'$ is a rank-$1$ matrix of the form $E \stackrel{\rm def} = A-A' = e_i (\vec v-\vec v')\T$. We thus denote $M$ as the matrix with the $i$-th row zeroed out, and have $M\T M  = A\T A - \vec v \vec v\T = A'\T A' - \vec v' \vec v'\T$. Recall that we assume that $\sigma_{\min}(A),\sigma_{\min}(A') \geq w$ and $\|E\| = \|\vec v-\vec v'\| \leq 2B$. We transpose $A$ and $R$ and denote $X = A\T R\T$ and $X' = (A')\T R\T$. For each column $\vec y_j$ of $R\T$ it holds that $\vec y_j\T \sim \gauss{\vec{0}_n, I_{n\times n}}$, and therefore the $j$-th column of $X$ is distributed like a random variable from $\gauss{\vec{0}_r, A\T A}$. Furthermore, as the columns of $R$ are independently chosen, so are the columns of $X$ are independent of one another. Therefore, for any $r$ vectors $\vec x_1, ...,\vec x_r\in \R^d$ it holds that
\begin{align*}
&\PDF_X( \vec x_1, ..., \vec x_r) = \prod_{j=1}^r \left({\sqrt{(2\pi)^d\det(A\T A)}}\right)^{-1} \exp\left(-\tfrac 1 2 \vec x_j\T (A\T A)^{-1} \vec x_j \right)\cr
&\PDF_{X'}( \vec x_1, ..., \vec x_r) = \prod_{j=1}^r \left( {\sqrt{(2\pi)^d\det(A'\T A')}}\right)^{-1} \exp\left(-\tfrac 1 2 \vec x_j\T (A'\T A')^{-1} \vec x_j \right)
\end{align*}
We apply the Matrix Determinant Lemma, and the Sherman-Morrison Lemma, and deduce:
\begin{align*}
&\det(A\T A) = \det(M\T M)\left(1+ \vec v\T (M\T M)^{-1} \vec v \right) \cr
&\det(A'\T A') = \det(M\T M)\left(1+ \vec v'\T (M\T M)^{-1} \vec v' \right)\cr
& (A\T A)^{-1} = (M\T M)^{-1} - \frac{ (M\T M)^{-1} \vec v \vec v\T (M\T M)^{-1}} {1+ \vec v\T (M\T M)^{-1} \vec v} \cr
& (A'\T A')^{-1} = (M\T M)^{-1} - \frac{ (M\T M)^{-1} \vec v' \vec v'\T (M\T M)^{-1}} {1+ \vec v'\T (M\T M)^{-1} \vec v'} 
\end{align*}
Together with the inequality $\frac {1+x}{1+y} = (1+x)(1-\frac y {1+y}) \leq \exp(x-\frac y {1-y})$ 
for any $x,y\neq 1$ 
we have
\begin{align}
\frac{\PDF_X( \vec x_1, ..., \vec x_r)} {\PDF_{X'}( \vec x_1, ..., \vec x_r)} & = \prod_{j=1}^r \sqrt{\frac {\det(A'\T A')} {\det(A\T A)}} \exp\left(-\tfrac 1 2 \vec x_j\T ((A\T A)^{-1} - (A'\T A')^{-1}) \vec x_j \right)\cr 
& = \prod_{j=1}^r \left( \frac {1 + \vec v'\T (M\T M)^{-1} \vec v'}{1 + \vec v\T (M\T M)^{-1} \vec v} \right)^{\tfrac 1 2} \exp\left(-\tfrac 1 2 \vec x_j\T ((A\T A)^{-1} - (A'\T A')^{-1}) \vec x_j \right)
\cr & \leq \prod_{j=1}^d \exp\Big( \frac 1 2 \left( \vec v'\T (M\T M)^{-1} \vec v'  - \frac{ \vec x_j\T(M\T M)^{-1} \vec v' \vec v'\T (M\T M)^{-1}\vec x_j} {1+ \vec v'\T (M\T M)^{-1} \vec v'}\right)
\cr &~~~~~~ + \frac 1 2 \left( - \frac {\vec v\T (M\T M)^{-1} \vec v}{1+\vec v\T (M\T M)^{-1} \vec v}  + \frac{ \vec x_j\T(M\T M)^{-1} \vec v \vec v\T (M\T M)^{-1}\vec x_j} {1+ \vec v\T (M\T M)^{-1} \vec v}\right)\Big)
\cr & = \exp \left( \frac 1 2 \left( r\cdot \vec v'\T (M\T M)^{-1} \vec v'  - \frac{ \vec v'\T(M\T M)^{-1} \left(\sum_{j=1}^r \vec x_j \vec x_j\T\right) (M\T M)^{-1}\vec v'} {1+ \vec v'\T (M\T M)^{-1} \vec v'}\right)\right)
\cr &~~~~~~ \cdot\exp\left(\frac 1 2 \left( - \frac {r\cdot \vec v\T (M\T M)^{-1} \vec v}{1+\vec v\T (M\T M)^{-1} \vec v}  + \frac{ \vec v\T(M\T M)^{-1} \left(\sum_{j=1}^r \vec x_j \vec x_j\T\right) (M\T M)^{-1}\vec v} {1+ \vec v\T (M\T M)^{-1} \vec v}\right)\right)
\label{eq:PDF_ratio}
\end{align}
Denote
\begin{align*}
z_1 &\stackrel{\rm def} = \vec v'\T (M\T M)^{-1} \vec v'  -  \vec v'\T(M\T M)^{-1} \left(\tfrac 1 r\sum_{j=1}^r \vec x_j \vec x_j\T\right) (M\T M)^{-1}\vec v'\cr
z_2 &\stackrel{\rm def} = \vec v\T (M\T M)^{-1} \vec v  -  \vec v\T(M\T M)^{-1} \left(\tfrac 1 r\sum_{j=1}^r \vec x_j \vec x_j\T\right) (M\T M)^{-1}\vec v\cr
\end{align*}
we have that
\begin{align*}
\ln\left( \frac{\PDF_X( \vec x_1, ..., \vec x_r)} {\PDF_{X'}( \vec x_1, ..., \vec x_r)}\right) &\leq \frac r 2\left( \frac {z_1} {1+\vec v' (M\T M)^{-1} \vec v'} + \frac {-z_2} {1+\vec v (M\T M)^{-1} \vec v} + \frac{(\vec v' (M\T M)^{-1} \vec v')^2} {1+\vec v' (M\T M)^{-1} \vec v'} \right)\cr
& \leq \frac r 2 \left( |z_1|+|z_2| + (\vec v' (M\T M)^{-1} \vec v')^2 \right)
\end{align*}

We now turn to analyze each of the above three terms separately.
The easiest to bound are the terms $\vec v (M\T M)^{-1} \vec v$ and $\vec v' (M\T M)^{-1} \vec v'$. Weyl's inequality yields that $\sigma_{\min}(M\T M) \geq \sigma_{\min}(A\T A) - B^2$, and so we give both terms that bound $\tfrac{B^2}{w^2-B^2} = \left(\tfrac{w^2}{B^2}-1 \right)^{-1}$. We turn to bounding $|z_1|$, $|z_2|$.

We continue assuming that $\vec x_1, \ldots, \vec x_r$ were sampled from $A\T A$. If they were sampled from $A'\T A'$ then the proof is analogous. Denote $X$ as the matrix whose columns are  $\vec x_1, \ldots, \vec x_r$. We have
\begin{align*}
z_2 & = ((M\T M)^{-1}\vec v)\T \left( M\T M - \left( \tfrac 1 r X\T X\right)  \right)(M\T M)^{-1}\vec v \cr
& = ((M\T M)^{-1}\vec v)\T \left( A\T A - \vec v\vec v\T - \left( \tfrac 1 r X\T X\right)  \right)(M\T M)^{-1}\vec v \cr
& = ((M\T M)^{-1}\vec v)\T (A\T A)^{1/2} \left( I - (A\T A)^{-1/2}\left( \tfrac 1 r X\T X\right)(A\T A)^{-1/2}  \right)(A\T A)^{1/2}(M\T M)^{-1}\vec v - (\vec v (M\T M)^{-1} \vec v)^2
\end{align*}

Recall that $X$ is a matrix whose rows are i.i.d samples from the multivariate Gaussian $\gauss{\vec 0, A\T A}$. Therefore, the rows of the matrix $X(A\T A)^{-1/2}$ are i.i.d samples from $\gauss{\vec 0, I_{d\times d}}$. In other words, the distribution of $X(A\T A)^{-1/2}$ is the same as a matrix whose entries are i.i.d samples from $\gauss{0,1}$. We can therefore invoke Lemma~\ref{lem:JL_of_Wishart_and_invWishart} and have that w.p. $\geq 1-\delta/2$.
\begin{align*}
|z_2| &\leq (2\sqrt{\tfrac {2\ln(4/\delta)}{r}}+\tfrac {2\ln(4/\delta)}{r} )\left \| (A\T A)^{1/2}(M\T M)^{-1}\vec v \right\|^2 + (\vec v (M\T M)^{-1} \vec v)^2 
\cr &\leq (2\sqrt{\tfrac {2\ln(4/\delta)}{r}}+\tfrac {2\ln(4/\delta)}{r} )\left( \vec v\T  (M\T M)^{-1}(M\T M + \vec v\vec v\T)(M\T M)^{-1}\vec v \right) +(\vec v (M\T M)^{-1} \vec v)^2 
\cr &= (\vec v (M\T M)^{-1} \vec v)\left(2\sqrt{\tfrac {2\ln(4/\delta)}{r}}+\tfrac {2\ln(4/\delta)}{r}\right)   + (\vec v (M\T M)^{-1} \vec v)^2  \left(2\sqrt{\tfrac {2\ln(4/\delta)}{r}}+\tfrac {2\ln(4/\delta)}{r}+1 \right)
\cr &\leq \left(\tfrac{w^2}{B^2}-1 \right)^{-1}\left(2\sqrt{\tfrac {2\ln(4/\delta)}{r}}+\tfrac {2\ln(4/\delta)}{r}\right)   + \left(\tfrac{w^2}{B^2}-1 \right)^{-2}  \left(2\sqrt{\tfrac {2\ln(4/\delta)}{r}}+\tfrac {2\ln(4/\delta)}{r}+1 \right)
\end{align*}

As the bound on $|z_1|$ is the same as the bound on $|z_2|$ we conclude that
\begin{align*}
\ln\left( \frac{\PDF_X( \vec x_1, ..., \vec x_r)} {\PDF_{X'}( \vec x_1, ..., \vec x_r)}\right) & \leq \frac r 2 \left( |z_1|+|z_2| + (\vec v' (M\T M)^{-1} \vec v')^2 \right)\cr
&\leq \left(\tfrac{w^2}{B^2}-1 \right)^{-1}\left(2\sqrt{{2r\ln(4/\delta)}}+ {2\ln(4/\delta)}\right)  \cr &~~~~+ \left(\tfrac{w^2}{B^2}-1 \right)^{-2}  \left(2\sqrt{{2r\ln(4/\delta)}}+ {2\ln(4/\delta)}+\frac {3r} 2 \right)\cr
&\leq \frac{\epsilon}{1+\tfrac \epsilon{\ln(4/\delta)}} +\epsilon^2\left(\frac{2\sqrt{{2r\ln(4/\delta)}}+ {2\ln(4/\delta)}}{(2\sqrt{{2r\ln(4/\delta)}}+ {2\ln(4/\delta)})^2} + \frac{3r}{16r\ln(4/\delta)}\right)\cr
&\leq \frac{\epsilon}{1+\tfrac \epsilon{\ln(4/\delta)}}\left(1 + \frac{\epsilon}{\ln(4/\delta)} \left(\frac 1 {2}+\frac {3}{16}\right)\right) < \epsilon
\end{align*}
by plugging in the value of $w^2$.
\end{proof}

\subsection{Privacy Proof for Algorithm~\ref{alg:additive_Wishart}}

\begin{theorem}
\label{thm:Wishart}
Fix $\epsilon\in (0, 1)$ and $\delta \in(0,\tfrac 1 e)$. Fix $B>0$. Let $C_1$ and $C_2$ be such that they satisfy
\[ \frac {2\sqrt{C_2}}{C_1(\sqrt{C_2}-1)^2} \leq \frac \epsilon {B^2}\] (E.g., $C_1 = B^2$ and $C_2 = \tfrac{14}{\epsilon^2}$.)
Let $A$ be a $(n\times d)$-matrix where each row of $A$ has bounded $L_2$-norm of $B$. Let $N$ be a matrix sampled from the $d$-dimensional Wishart distribution with $\nu$-degrees of freedom using the scale matrix $V$ (i.e., $N\sim \W_d(V,\nu)$) for any matrix $V$ with least singular value $\sigma_{\min}(V) \geq C_1$ (e.g. $V= C_1 I_{d\times d}$) and $\nu \geq \lfloor d+2C_2\ln(4/\delta)\rfloor$. 
Then outputting $X = A\T A + N$ is $(\epsilon,\delta)$-differentially private.
\end{theorem}
We comment that in order to sample such an $N$, one can sample a matrix $N' \in \mathbb{R}^{\nu\times d}$ of i.i.d normal Gaussians, multiple all entries by $B/\sqrt{\epsilon}$ and set $N' = N\T N$.
\begin{proof}
Fix $A$ and $A'$ that are two neighboring datasets that differ on the $i$-th row, denoted as $\vec v\T$ in $A$ and $\vec v'\T$ in $A'$. Let $M$ denote $A$ or $A'$ without the $i$-th row, i.e. $M\T M = A\T A - \vec v \vec v\T = A'\T A' - \vec v' \vec v'\T$. Therefore, denoting $\sigma_{\min}(M)$ and $\sigma_{\min}(A)$ as the least singular value of $M$ and $A$ resp., we have that  $\sigma_{\min}^2(M) \leq \sigma_{\min}^2(A) \leq \sigma_{\min}^2(M)+B^2$. Same holds for the least singular value of $M$ and $A'$.

Recall that 
\[ \PDF_{\W_d(V,\nu)}(N) \propto \det(N)^{\tfrac{\nu-d-1}2} \exp\left(-\tfrac 1 2\trace(V^{-1}N)\right)\]
We argue that Wishart-matrix additive noise is $(\epsilon,\delta)$-differentially private, using the explicit formulation of the $\PDF$. For the time being, we ignore the issue of outputting a matrix $X$ s.t. either $X-A\T A$, $X-A'\T A'$ or $X-M\T M$ are non-invertible. (Note, if our input matrix is $A$, then $\Pr[X-A\T A \textrm{ non invertible}] = \Pr_{N\sim \W_d(V,\nu)}[N\textrm{ non invertible}] = 0$. However, it is not a-priori clear why we should also have $\Pr[X-A'\T A' \textrm{ non invertible}] =0$ or $\Pr[X-M\T M \textrm{ non invertible}] =0$.) Later, we justify why such events can be ignored. We now bound the appropriate ratios. If we denote the output of the mechanism as a matrix $X$, then we compare

\begin{align*}
\frac {\PDF_{\W_d(V,\nu)}(X-A\T A)} {\PDF_{W_d(V,\nu)}(X-A'\T A')} &= \left( \frac {\det(X-A\T A)}{\det(X-A'\T A')} \right)^{\tfrac {\nu-d-1}2} e^{-\tfrac 1 2 \left( \trace(V^{-1}(X-A\T A))-\trace(V^{-1}(X-A'\T A') \right)} \cr
& =\left( \frac {\det(X-M\T M - \vec v \vec v\T)}{\det(X-M\T M - \vec v' \vec v'\T)} \right)^{\tfrac {\nu-d-1}2} e^{-\tfrac 1 2 \left( \trace(V^{-1}(X-A\T A-X+A'\T A'\T)\right)} \cr
& = \left(\frac{1 - \vec v\T (X-M\T M)^{-1}\vec v} {1 - \vec v'\T (X-M\T M)^{-1}\vec v'} \right)^{\tfrac {\nu-d-1}2} \exp\left(-\tfrac 1 2 \trace(V^{-1} \vec v'\vec v'\T) + \tfrac 1 2 \trace(V^{-1}\vec v \vec v\T) \right)\cr
& \stackrel{\trace(AB)=\trace(BA)}=   \left(\frac{1 - \vec v\T (X-M\T M)^{-1}\vec v} {1 - \vec v'\T (X-M\T M)^{-1}\vec v'} \right)^{\tfrac {\nu-d-1}2} \exp\left(-\tfrac 1 2 \vec v'\T V^{-1} \vec v' + \tfrac 1 2 \vec v\T V^{-1}\vec v \right) 
\end{align*}
We can now use the inequality $\frac {1-x}{1-y} = (1 -x)(1+ \frac{y}{1-y}) \leq \exp(-x+\frac {y}{1-y})$ for any $x$ and any $y\neq 1$ to deduce
\begin{align*}
\ln\left(\frac {\PDF_{A\T A + N}(X)} {\PDF_{A'\T A' + N}(X)}\right) & \leq \frac 1 2\cdot \vec v\T\left( V^{-1} - (\nu-d-1) (X-M\T M)^{-1}\right)\vec v \cr
&+ \frac 1 2\cdot \vec v'\T\left( \frac {\nu-d-1} {1 - \vec v'\T (X-M\T M)^{-1}\vec v'} (X-M\T M)^{-1}-V^{-1}\right)\vec v' \cr
\end{align*}
Note that we either have $X-M\T M = X-A\T A + \vec v\vec v\T = N +\vec v \vec v\T$ or $X-M\T M = N + \vec v' \vec v'\T$. And so, we continue assuming $X$ was sampled using $A\T A$, but the case $X$ was sampled from $A'\T A'$ is symmetric. Further, we only show a bound for the first term of the two above, as the other term will have the same upper bound.

Note that $(X-M\T M)^{-1} = (X-A\T A + \vec v \vec v\T)^{-1}= (X-A\T A)^{-1} - \frac{(X-A\T A)^{-1} \vec v \vec v\T (X-A\T A)^{-1} }{1+ \vec v (X-A\T A)^{-1} \vec v}$, hence
\begin{align*}
\vec v\T (X-M\T M)^{-1} \vec v & = \vec v\T (X-A\T A)^{-1} \vec v - \frac {(\vec v\T (X-A\T A)^{-1} \vec v)^2}{1+\vec v\T (X-A\T A)^{-1} \vec v} &= \frac {\vec v\T (X-A\T A)^{-1} \vec v}{1+\vec v\T (X-A\T A)^{-1} \vec v}
\cr \vec v'\T (X-M\T M)^{-1} \vec v' & = \vec v'\T (X-A\T A)^{-1} \vec v' - \frac {(\vec v'\T (X-A\T A)^{-1} \vec v)^2}{1+\vec v\T (X-A\T A)^{-1} \vec v}\end{align*}

And so we have:
\begin{align*}
& \vec v\T\left( V^{-1} - (\nu-d-1) (X-M\T M)^{-1}\right)\vec v 
\cr &~~~ = \vec v\T\left( V^{-1} - (\nu-d+1) (X-M\T M)^{-1}\right)\vec v + 2\vec v\T (X-M\T M)^{-1}\vec v
\cr &~~~ \leq \vec v\T\left( V^{-1} - (\nu-d+1) (X-A\T A)^{-1}\right)\vec v + 2\vec v\T (X-A\T A)^{-1}\vec v + (\nu-d+1)(\vec v\T ( X-A\T A )^{-1} \vec v)^2
\end{align*}

Now, note that $(X-A\T A)\sim \W_d(V,\nu)$, and so $V^{-1/2}(X-A\T A)V^{-1/2} \sim \W_d(I_{d\times d},\nu)$. This allows us to invoke Lemma~\ref{lem:JL_of_Wishart_and_invWishart} to \[\vec v\T\left(V^{-1} - (\tfrac 1 {\nu-d+1}(X-A\T A))^{-1}\right)\vec v = (V^{-1/2}\vec v)\T\left( I -  \left( \frac{V^{-1/2}(X-A\T A)V^{-1/2}}{\nu-d+1}\right)^{-1}\right) (V^{-1/2}\vec v)\] and infer that w.p. $\geq 1-\delta/2$ we have the following bound 
\newcommand{\teq}{\sqrt{\tfrac{2\ln(4/\delta)}{\nu-d+1}}}
\begin{align*}
& \vec v\T\left( V^{-1} - (\nu-d-1) (X-M\T M)^{-1}\right)\vec v 
\cr
& ~~~ \leq \left(\frac {2\sqrt{2(\scriptstyle{\nu-d+1})\ln(4/\delta)} - 2\ln(4/\delta)}{(\sqrt{\scriptstyle \nu-d+1}-\sqrt{\scriptstyle 2\ln(4/\delta)})^2} + \frac {2} {(\sqrt{\scriptstyle \nu-d+1}-\sqrt{\scriptstyle 2\ln(4/\delta)})^2} + \frac {2(\nu-d+1)} {(\sqrt{\scriptstyle \nu-d+1}-\sqrt{\scriptstyle 2\ln(4/\delta)})^4}\right) \|V^{-1/2}\vec v\|^2
\cr & ~~~ = \frac {\|V^{-1/2}\vec v\|^2}{ {(\sqrt{\scriptstyle \nu-d+1}-\sqrt{\scriptstyle 2\ln(4/\delta)})^2} }\left( 2\sqrt{2(\nu-d+1)\ln(4/\delta)} - 2\ln(4/\delta) + 2 +\frac 2 {(1-\sqrt{\frac{2\ln(4/\delta)}{\nu-d-1}})^2} \right)
\cr & ~~~ = \frac {2\sqrt{2(\nu-d+1)\ln(4/\delta)} - 2\ln(4/\delta) + 6}{(\sqrt{\scriptstyle \nu-d+1}-\sqrt{\scriptstyle 2\ln(4/\delta)})^2} \|V^{-1/2}\vec v\|^2
\end{align*}
Analogously, w.p. $\geq 1-\delta/2$ the following bound holds as well:
\begin{align*}
&\vec v'\T\left( \frac {\nu-d-1} {1 - \vec v'\T (X-M\T M)^{-1}\vec v'} (X-M\T M)^{-1}-V^{-1}\right)\vec v' 
\cr & ~~~~ = \vec v'\T\left( (\nu-d-1) (X-M\T M)^{-1}-V^{-1}\right)\vec v' + \frac{(\nu-d-1)(\vec v'\T (X-M\T M)^{-1}\vec v')^2} {1 - \vec v'\T (X-M\T M)^{-1}\vec v'}
\cr & ~~~~ \leq \vec v'\T\left( (\nu-d+1) (X-M\T M)^{-1}-V^{-1}\right)\vec v' + \frac{(\nu-d-1)(\vec v'\T (X-M\T M)^{-1}\vec v')^2} {1 - \vec v'\T (X-M\T M)^{-1}\vec v'}
\cr & ~~~~ \leq \vec v'\T\left( (\nu-d+1) (X-A\T A)^{-1}-V^{-1}\right)\vec v' + \frac{(\nu-d-1)(\vec v'\T (X-A\T A)^{-1}\vec v')^2} {1 - \vec v'\T (X-M\T M)^{-1}\vec v'}
\cr & ~~~~ \leq \left(\frac {2\sqrt{2(\scriptstyle{\nu-d+1})\ln(4/\delta)} - 2\ln(4/\delta)}{(\sqrt{\scriptstyle \nu-d+1}-\sqrt{\scriptstyle 2\ln(4/\delta)})^2} + \frac {2(\nu-d+1)} {(\sqrt{\scriptstyle \nu-d+1}-\sqrt{\scriptstyle 2\ln(4/\delta)})^4}\right) \|V^{-1/2}\vec v'\|^2
\end{align*}

Combining the two upper bounds we get
\begin{align*}
\ln\left(\frac {\PDF_{A\T A + N}(X)} {\PDF_{A'\T A' + N}(X)}\right) & \leq \frac 1 2\cdot \vec v\T\left( V^{-1} - \frac {\nu-d-1} {1 - \vec v'\T (X-M\T M)^{-1}\vec v'} (X-M\T M)^{-1}\right)\vec v \cr
&+ \frac 1 2\cdot \vec v'\T\left( \frac {\nu-d-1} {1 - \vec v'\T (X-M\T M)^{-1}\vec v'} (X-M\T M)^{-1}-V^{-1}\right)\vec v' 
\cr & \leq \frac {2\sqrt{2(\nu-d+1)\ln(4/\delta)} - 2\ln(4/\delta) + 6}{(\sqrt{\scriptstyle \nu-d+1}-\sqrt{\scriptstyle 2\ln(4/\delta)})^2} \cdot \frac {\|V^{-1/2}\vec v\|^2+\|V^{-1/2}\vec v'\|^2} 2
\cr & \stackrel{\delta< \tfrac 1 6} \leq \frac {B^2}{\sigma_{\min}(V)} \cdot \frac {2\sqrt{2(\nu-d+1)\ln(4/\delta)}}{(\sqrt{\scriptstyle \nu-d+1}-\sqrt{\scriptstyle 2\ln(4/\delta)})^2}
\end{align*}

All we now need to do is to plug in the fact that $\nu = \lfloor d + C_2\cdot 2\ln(4/\delta)\rfloor \geq d-1+ C_2\cdot 2\ln(4/\delta)$, and that $\sigma_{\min}(V) \geq C_1$ to deduce 
\begin{align*}
\ln\left(\frac {\PDF_{A\T A + N}(X)} {\PDF_{A'\T A' + N}(X)}\right) & \leq \frac {B^2}{C_1} \cdot \frac {2\cdot 2\ln(4/\delta)\cdot \sqrt{C_2}}{(\sqrt{C_2\cdot 2\ln(4/\delta)}-\sqrt{2\ln(4/\delta)})^2} \leq \frac {2B^2\sqrt{C_2}}{C_1 (\sqrt{C_2}-1)^2} \leq \epsilon
\end{align*}
\end{proof}

\subsection{Privacy Proof for Algorithm~\ref{alg:inv_Wishart}}

\begin{theorem}
\label{thm:inverse_Wishart}
Fix $\epsilon>0$ and $\delta \in(0,\tfrac 1 e)$. Fix $B>0$. 
Let $A$ be a $(n\times d)$-matrix and fix an integer $\nu \geq d$. 
Let $w$ be such that 
\[ w^2= \frac {B^2}{\epsilon(1-\tfrac \epsilon {2\ln(4/\delta)})} \left(2\sqrt{2\nu\ln(4/\delta)} + 2\ln(4/\delta) \right)\]
Then, given that $\sigma_{\min}(A) \geq w$, the algorithm that samples a matrix from $\W_d^{-1}(A\T A,\nu)$ is $(\epsilon,\delta)$-differentially private.
\end{theorem}

We comment on the similarity between the bounds of Theorem~\ref{thm:JL} and Theorem~\ref{thm:inverse_Wishart}. This is after all quite natural, since the JL-theorem is a way to sample from a Wishart distribution $\W_d(A\T A, r)$ ( since every row in the matrix $RA$ is an i.i.d sample from $\gauss{\vec 0, A\T A}$). Clearly, one can sample a matrix from $\W_d(A\T A, r)$ and invert it, to get a sample from $\W_d^{-1}((A\T A)^{-1}, r)$ and vice-versa. Therefore, we get similar bounds. The only slight difference lies in the fact that we require in Theorem~\ref{thm:inverse_Wishart} that $\nu \geq d$, s.t. the matrix we sample is indeed invertible, whereas we do not require any such lower bound for sampling from $\W_d(A\T A,r)$.

\begin{proof}
As always, we denote $A'$ as a neighbor of $A$ that differs just on a single row, which we denote $\vec v$ for $A$ and $\vec v'$ for $A'$, and as before, the matrix $M$ is the matrix $A$ with the $i$-th row all zeroed out. Therefore, $A\T A - \vec v\vec v\T = A'\T A' - \vec v' \vec v'\T = M\T M$.
So, denoting $\sigma_{\min}(M)$ and $\sigma_{\min}(A)$ as the least singular value of $M$ and $A$ resp., we have that  $\sigma_{\min}^2(M) \leq \sigma_{\min}^2(A) \leq \sigma_{\min}^2(M)+B^2$. Same holds for the least singular value of $M$ and $A'$.

Recall that 
\[ \PDF_{\W^{-1}_d(A\T A,\nu)}(X) \propto \det(A\T A)^{\tfrac{\nu}2}\det(X)^{\tfrac{\nu+p+1}2} \exp\left(-\tfrac 1 2\trace((A\T A)X^{-1})\right)\]
We invoke the determinant update lemma, the Sherman Morisson lemma and the inequality $\tfrac {1+x}{1+y} \leq \exp({x - \tfrac y{1+y}})$ yet again to deduce:
\begin{align*}
 \frac{ \PDF_{\W^{-1}_d(A\T A,\nu)}(X) } {\PDF_{\W^{-1}_d(A'\T A',\nu)}(X)} &= \frac{\det(A\T A)^{\nu/2}\exp\left(-\tfrac 1 2 \trace( (A\T A) X^{-1} ) \right)} {\det(A'\T A')^{\nu/2}\exp\left(-\tfrac 1 2 \trace( (A'\T A') X^{-1} ) \right)} 
\cr &= \left(\frac{1+\vec v\T (M\T M)^{-1} \vec v} {1+\vec v'\T (M\T M)^{-1} \vec v'}\right)^{\nu/2} \exp\left(-\frac 1 2 \trace( (A\T A-A'\T A') X^{-1} \right)
\cr & \leq \exp\left(  \frac \nu 2\left( \vec v\T(M\T M)^{-1} \vec v- \frac{\vec v'\T(M\T M)^{-1} \vec v'} {1+\vec v'\T(M\T M)^{-1} \vec v'}\right)  -\frac 1 2 \left(\trace((\vec v \vec v\T - \vec v'\vec v'\T) X^{-1})\right)  \right)
\cr & = \exp\left( \frac 1 2 \left({\nu\cdot \vec v\T(M\T M)^{-1} \vec v}  - \vec v\T X^{-1}\vec v\right) - \frac 1 2\left(\frac{\nu\cdot \vec v'\T(M\T M)^{-1} \vec v'} {1+\vec v'\T(M\T M)^{-1} \vec v'} - \vec v'\T X^{-1}\vec v'\right) \right)
\cr &\leq \exp\left( \frac 1 2 \vec v\T\left(\nu(M\T M)^{-1}- X^{-1}\right)\vec v - \frac 1 2 {\vec v'\T\left(\frac{\nu}{1+\vec v'\T(M\T M)^{-1} \vec v'}(M\T M)^{-1} - X^{-1}\right)\vec v'}   \right)
\end{align*}
We continue assuming $X \sim \W_d^{-1}(A\T A, \nu)$ (the case $X\sim \W_d^{-1}(A'\T A', \nu)$ is symmetric). By definition, we have that $X^{-1} \sim \W_d((A\T A)^{-1}, \nu)$. Hence $(A\T A)^{1/2} X^{-1} (A\T A)^{-1/2} \sim \W_d(I_{d\times d}, \nu)$, which implies that the distribution of $(A\T A)^{1/2} X^{-1} (A\T A)^{-1/2}$ is the same as generating a $(\nu\times d)$-matrix of i.i.d $\gauss{0,1}$ samples and take its cross-product with itself.

We continue using the Sherman-Morrison formula, and derive the bound
\begin{align*}
\vec v\T\left(\nu(M\T M)^{-1}- X^{-1}\right)\vec v &= \vec v\T\left(\nu(A\T A)^{-1}- X^{-1}\right)\vec v - \frac{\nu \cdot (\vec v\T (A\T A)^{-1} \vec v)^2}{1-\vec v\T (A\T A)^{-1} \vec v}
\cr &\leq ((A\T A)^{-1/2}\vec v)\T\left(\nu I_{d\times d}- (A\T A)^{1/2}X^{-1}(A\T A)^{1/2}\right)((A\T A)^{-1/2}\vec v) 
\cr &\leq \|(A\T A)^{-1/2}\vec v\|^2 \left( 2\sqrt{2\nu\ln(4/\delta)} + 2\ln(4/\delta)\right) 
\end{align*}
which holds w.p. $\geq 1- \delta/2$ due to Lemma~\ref{lem:JL_of_Wishart_and_invWishart}. Similarly, we have
\begin{align*}
&-\vec v'\T\left(\frac{\nu}{1+\vec v'\T(M\T M)^{-1} \vec v'}(M\T M)^{-1} - X^{-1}\right)\vec v' 
\cr &~~~~~~= -\vec v'\T \left( \nu(M\T M)^{-1} - X^{-1} \right) \vec v' + \frac{\nu \cdot (\vec v'\T (M\T M)^{-1} \vec v')^2}{1-\vec v'\T (M\T M)^{-1} \vec v'}
\cr &~~~~~~= -\vec v'\T \left( \nu(A\T A)^{-1} - X^{-1} \right) \vec v' + \frac{\nu \cdot (\vec v'\T (M\T M)^{-1} \vec v')^2}{1-\vec v'\T (M\T M)^{-1} \vec v'}+\frac{\nu \cdot (\vec v'\T (A\T A)^{-1} \vec v)^2}{1-\vec v'\T (A\T A)^{-1} \vec v'}
\cr &~~~~~~\leq -\vec v'\T \left( \nu(A\T A)^{-1} - X^{-1} \right) \vec v' + \frac{\nu \cdot (\vec v'\T (A\T A)^{-1} \vec v)^2}{1-\vec v'\T (M\T M)^{-1} \vec v'}+\frac{\nu \cdot (\vec v'\T (A\T A)^{-1} \vec v)^2}{1-\vec v'\T (A\T A)^{-1} \vec v'}
\cr &~~~~~~\leq \|(A\T A)^{-1/2}\vec v'\|^2 \left(2\sqrt{2\nu\ln(4/\delta)} + 2\ln(4/\delta)\right) 
\cr &~~~~~~~~~+ \nu\cdot \|(A\T A)^{-1/2}\vec v'\|^2\|(A\T A)^{-1/2}\vec v\|^2\left(\frac 1 {1-\vec v'\T (M\T M)^{-1} \vec v'} + \frac 1 {1-\vec v'\T (A\T A)^{-1} \vec v'} \right)
\end{align*}
Denoting the least singular value of $(A\T A)$ as $w^2$, and using the fact that $\|\vec v\|,\|\vec v'\|\leq B$ and crudely upper bounding $\vec v'\T (M\T M)^{-1} \vec v'$ and $\vec v'\T (A\T A)^{-1} \vec v'$ by $\tfrac 1 2$ we get
\begin{align*}
\ln\left( \frac{ \PDF_{\W^{-1}_d(A\T A,\nu)}(X) } {\PDF_{\W^{-1}_d(A'\T A',\nu)}(X)}  \right) &\leq\frac 1 2 \cdot 2\cdot \frac {B^2}{w^2} \left(2\sqrt{2\nu\ln(4/\delta)} + 2\ln(4/\delta)\right) + \frac 1 2 \cdot \frac {B^4}{w^4}(4\nu+4\nu) 
\end{align*}
As we have $w^2 = \frac {B^2}{\epsilon(1-\tfrac \epsilon {2\ln(4/\delta)})} \left(2\sqrt{2\nu\ln(4/\delta)} + 2\ln(4/\delta) \right)$ we get that 
\begin{align*}
\ln\left( \frac{ \PDF_{\W^{-1}_d(A\T A,\nu)}(X) } {\PDF_{\W^{-1}_d(A'\T A',\nu)}(X)}  \right) &\leq \epsilon(1-\tfrac \epsilon {2\ln(4/\delta)}) + \epsilon^2 \frac {4\nu}{8\nu\ln(4/\delta)} \leq \epsilon
\end{align*}

\end{proof}

\subsection{An Additional Privacy Theorem --- Gaussian Noise for the Inverse}

\begin{theorem}
\label{thm:Gaussian_noise_on_inverse}
Fix $\epsilon \in (0,1)$ and $\delta \in (0,e^{-1})$. Let $A$ be a $(n\times d)$-matrix where the $l_2$-norm of each row is bounded by $B$, where it is publicly known that $\frac{\sigma_{\min}(A\T A)}{B^2} \geq 1+\rho$ with $\rho > 0$. Then the algorithm that outputs $(A\T A)^{-1} + N$ where $N$ is a symmetric matrix with each entry on or above the main diagonal of $N$ is sampled i.i.d from $\gauss{0, \frac {8\log(2/\delta)} {\rho^2\epsilon^2}}$ is $(\epsilon,\delta)$-differentially private.
\end{theorem}
\begin{proof}
The proof of the theorem just bounds the $l_2$-global sensitivity of the inverse, using the Sherman Morrison formula. We then use the fact that by independently adding noise to each entry in $(A\T A)^{-1}_{j,k}$ for $j\leq k$ where the noise is sampled i.i.d from $\gauss{0, GS_2^2\cdot  \frac {2\log(2/\delta)}{\epsilon^2}}$ is $(\epsilon,\delta)$-differentially private.

Denote $A$ and $A'$, two matrices that differ on a single row, which is denoted $\vec v$ in $A$ and $\vec v'$ in $A$. Therefore, $A'\T A' = A\T A + \vec v' \vec v'\T - \vec v \vec v\T$, so Weyl's inequality gives that $|\sigma_{\min} (A\T A) - \sigma_{\min}(A'\T A')| \leq~B^2$. Denoting $M$ as the matrix we get by zeroing out the $i$-th row on $A$ or $A'$, we have
\begin{align*}
(A\T A)^{-1} &= (M\T M)^{-1} - \frac {(M\T M)^{-1} \vec v \vec v\T (M\T M)^{-1}}{1+ \vec v\T (M\T M)^{-1} \vec v} \cr
(A'\T A')^{-1} &= (M\T M)^{-1} - \frac {(M\T M)^{-1} \vec v' \vec v'\T (M\T M)^{-1}}{1+ \vec v'\T (M\T M)^{-1} \vec v'} 
\end{align*}
Hence, \[(A'\T A')^{-1} - (A\T A)^{-1} = (M\T M)^{-1} \left( \frac {\vec v \vec v\T}{1+\vec v\T (M\T M)^{-1} \vec v} -  \frac {\vec v' \vec v'\T}{1+\vec v'\T (M\T M)^{-1} \vec v'} \right) (M\T M)^{-1}\]
Let $\vec x_1, \vec x_2, \ldots, \vec x_d$ be the eigenvectors of $M$, corresponding to the eigenvalues $\mu_1, \ldots, \mu_d$. Then, for any $j,k$ we have
\begin{align*}
\vec x_j\T \left( (A'\T A')^{-1} - (A\T A)^{-1} \right) \vec x_k &= \mu_j^{-1}\mu_k^{-1} \left( \frac {(\vec v \cdot \vec x_j)(\vec v\cdot \vec x_k)}{1+\vec v\T (M\T M)^{-1} \vec v} -  \frac {(\vec v'\cdot \vec x_j)(\vec v'\cdot \vec x_k)}{1+\vec v'\T (M\T M)^{-1} \vec v'} \right)
\end{align*}
Due to Weyl's inequality, $\sigma_{\min}(M\T M) \geq \sigma_{\min}(A\T A) -  \|\vec v\|^2 \geq \rho\cdot B^2$. And so, together with the inequality $(a-b)^2 \leq 2a^2 +2b^2$ we get
\begin{align*}
\left\| (A'\T A')^{-1} - (A\T A)^{-1}\right\|_F^2 &= \sum_{j,k=1}^d (\vec x_j\T \left( (A'\T A')^{-1} - (A\T A)^{-1} \right) \vec x_k)^2
\cr &= \frac 2 {1+\vec v\T (M\T M)^{-1} \vec v} \sum_{j,k} \frac{(\vec v\cdot \vec x_j)^2  (\vec v \cdot x_k)^2}{\mu_j^2\mu_k^2}
\cr & ~~~+  \frac 2 {1+\vec v'\T (M\T M)^{-1} \vec v'} \sum_{j,k} \frac{ (\vec v'\cdot \vec x_j)^2 (\vec v' \cdot x_k)^2}{\mu_j^2\mu_k^2}
\cr & \leq \frac 2 {(\rho B^2)^2} \sum_{j,k} (\vec v\cdot \vec x_j)^2  (\vec v \cdot x_k)^2 + (\vec v'\cdot \vec x_j)^2  (\vec v' \cdot x_k)^2
\cr & = \frac 2 {(\rho B^2)^2} \sum_{j,k} \left( \left(\sum_j (\vec v\cdot x_j)^2 \right)\left(\sum_k (\vec v\cdot x_k)^2 \right) + \left(\sum_j (\vec v'\cdot x_j)^2 \right)\left(\sum_k (\vec v'\cdot x_k)^2 \right)\right)
\cr & = \frac {2\|v\|^4 + 2\|v'\|^4}{(\rho B^2)^2} \leq \frac {4B^4}{(\rho B^2)^2} = \frac 4 {\rho^2}\qedhere
\end{align*}
\end{proof}

\section{Utility Theorems}
\label{sec:utility_analysis}

\cut{
We will use the fact that for symmetric $A$, $D$ and matrix $B$ whose dimensions match, we have
\begin{align*}
\left(\begin{array}{c c}
A & B \cr B\T & C
\end{array}\right)^{-1} &=  \left(\begin{array}{c c} A^{-1} & 0 \cr 0& 0
\end{array}\right) + \left(\begin{array}{c} -A^{-1}B \cr I
\end{array}\right) \left(\begin{array}{c} C-B\T A^{-1} B
\end{array} \right)^{-1}\left(\begin{array}{c c} -B\T A^{-1} & I
\end{array}\right) \cr
& = \left( \begin{array}{c c}
A^{-1} + A^{-1} B(C-B\T A^{-1}B)^{-1} B\T A^{-1} & -A^{-1} B (C-B\T A^{-1} B)^{-1} \cr
-(C-B\T A^{-1} B)^{-1} B\T A^{-1} & (C-B\T A^{-1} B)^{-1}
\end{array}\right)
\end{align*}
}
In this section we provide the utility statement for the Analyze Gauss algorithm and the additive Wishart noise algorithm.
Throughout this section we assume our database $D \in\mathbb{R}^{n\times d}$ is in fact composed of $D=[X;\vec y]$ where $X \in \mathbb{R}^{n\times p}$ and $\vec y\in \mathbb{R}^n$ (so we denote $p=d-1$). Clearly, to assume $\vec y$ is the last column of $D$ simplifies the notation, but $\vec y$ can be any single column of $D$ and $X$ can be any subset of the other columns of $D$.

In this section we will repeatedly use the Woodbury formula, which states that for any invertible $A \in \mathbb{R}^{p\times p}$ and $U \in \mathbb{R}^{p\times k}$ and $V\in\mathbb{R}^{k\times p}$ of corresponding dimension we have
\[ (A+U V)^{-1} = A^{-1} - A^{-1} U \left( I_{k\times k} - V A^{-1} U \right)^{-1} V A^{-1}\]
which implies that for any $B\in\mathbb{R}^{p\times p}$ we have the binomial inverse formula:
\begin{align} (A+B)^{-1} &= A^{-1} - A^{-1}(I_{p\times p} - BA^{-1})^{-1} B A^{-1} \label{eq:binomial_inverse_formula}
\end{align}

Our goal is to compare the distance between our predictor to the predictor one gets without noise, i.e. to $\hvec\beta = (X\T X)^{-1} X\T \vec y$. Since we release a matrix $\widetilde{D\T D}$ that approximates $D\T D$, we can decompose it into the $p\times p$ matrix $\widetilde{X\T X}$ and the $p$-dimensional vector $\widetilde{X\T \vec y}$ and compute $\widetilde{\vec \beta}=(\widetilde{X\T X})^{-1} \widetilde{X\T \vec y}$. We thus give bounds on
\[\left\|\widetilde{\vec\beta}-\hvec\beta\right\| = \left\|(\widetilde{X\T X})^{-1} \widetilde{X\T \vec y} - (X\T X)^{-1} X\T \vec y\right\| \]

Our analysis presents utility analysis that depends on the input parameters. This is in contrast to previous works on DP ERM that give a uniform bound and obtain it via \emph{regularization} of the problem. (This is natural, as for $X = 0_{n\times p}$ clearly $\hvec\beta$ is ill-defined unless we regularize the problem.)

\begin{theorem}
\label{thm:utility_analyze_gauss}
Fix $X\in\mathbb{R}^{n\times p}$ and $\vec y\in\mathbb{R}^n$ s.t. $X\T X$ is invertible. Fix $\eta \in (0,1)$ and $\nu \in (0,1/e)$. 
Denote $\widetilde{X\T X} = X\T X + N$ and $\widetilde{X\T \vec y} = X\T \vec y + \vec n$ where each entry of $N$ and $\vec n$ is sampled i.i.d from $\gauss{0,\sigma^2}$. Then, there exists some constant $C\geq 1$ s.t. if we have that $\sigma_{\min}(X\T X) \geq \tfrac {2C}\eta \cdot \sigma \sqrt{p} \log(1/\nu)$, then w.p. $\geq 1-\nu$ we have
\[\left\|\widetilde{\vec\beta}-\hvec\beta \right\| = \left\|(\widetilde{X\T X})^{-1} \widetilde{X\T \vec y} - (X\T X)^{-1} X\T \vec y\right\| \leq 2\eta \hvec\beta + \frac {\eta}C\]
\end{theorem}
We comment that this is not precisely the same as the behavior of the ``Analyze Gauss'' algorithm. The difference lies in the fact that Analyze Gauss outputs $X\T X + M$ where $M$ is a symmetric matrix whose entries along and above the main diagonal are sampled i.i.d from a suitable $\gauss{0,\sigma^2}$. However, one can denote $M = \tfrac 1 {\sqrt 2} (N + N\T)$ for a matrix $N$ whose entries are i.i.d samples from $\gauss{0,\sigma^2}$, and so the same result, up to a factor of $\sqrt{2}$, holds for Analyze Gauss. 
\begin{proof}
Plugging in \eqref{eq:binomial_inverse_formula} we get
\begin{align*}
(\widetilde{X\T X})^{-1} \widetilde{X\T \vec y} &= \left( I_{p\times p} - (X\T X)^{-1}(I_{p \times p} - N(X\T X)^{-1})^{-1} N \right)(X\T X)^{-1} X\T \vec y \cr &~~~ + \left( I_{p\times p} - (X\T X)^{-1}(I_{p \times p} - N(X\T X)^{-1})^{-1} N \right)(X\T X)^{-1} \vec n
\end{align*}
Denoting $Z = (X\T X)^{-1}(I_{p \times p} - N(X\T X)^{-1})^{-1} N$, we derive a bound on $\left\|(\widetilde{X\T X})^{-1} \widetilde{X\T \vec y} - (X\T X)^{-1} X\T \vec y\right\|$ using bounds on $\|Z\|$, $\|I-Z\|$ and $\|\vec n\|$.

Standard bounds on a symmetric ensemble of Gaussians~\cite{Tao12} give that $\|N\| \leq C\cdot \sigma\sqrt p \log(1/\nu)$ w.p. $\geq 1-\tfrac \nu 2$ for some suitable constant $C>0$. Hence we have that $\|N\|\cdot \|(X\T X)^{-1}\| \leq \eta$. Hence, all singular values of $N(X\T X)^{-1}$ are upper bounded in absolute value by $\eta$, and so all singular values of $I-N(X\T X)^{-1}$ lie in the range $[1-\eta,1+\eta]$. This implies that $\|Z \| \leq \tfrac \eta{1-\eta}$ and $\|I - Z\| \leq 1 +\tfrac \eta{1-\eta} = \tfrac 1 {1-\eta}$. Next we note that $\|\vec n\|^2 \sim \sigma^2\cdot \chi^2_p$, and so, w.p. $\geq 1 - \tfrac\nu 2$ it holds that $\|\vec n\| \leq \sigma (\sqrt p + \sqrt{2\ln(2/\nu)})$. 

Thus, we get
\begin{align*}
\left\|\tvec\beta-\hvec\beta\right\| &\leq \frac \eta{1-\eta} \|\hvec\beta\| + \frac 1 {1-\eta} \cdot \frac { \sqrt{\sigma^2 p}+\sqrt{2\sigma^2 \ln(2/\nu)} } {\sigma_{\min}(X\T X)} \leq \frac {\eta} {1-\eta} \|\hvec\beta\| + \frac \eta C
\end{align*}
\end{proof}

\begin{corollary}
\label{cor:analyze_gauss_distance}
Denote $\rho = \frac{\sigma_{\min}(\widetilde{X\T X})}{2\sigma \sqrt{p}\log(1/\nu)}$. Then, for the same constant $C$ in Theorem~\ref{thm:utility_analyze_gauss}, if $\rho \geq 2C$ then we have
\[ \left\| \tvec\beta-\hvec\beta\right\| \leq \frac {2C}{\rho}\|\tvec\beta\| + \frac 1 \rho \]
\end{corollary}
\begin{proof}
The proof follows from Theorem~\ref{thm:utility_analyze_gauss}, and the observation that we can flip the role of $X\T X$ and $\widetilde{X\T X}$ because the Gaussian distribution is symmetric. And so, we just use the notation $\rho = \tfrac C \eta$.
\end{proof}

\begin{theorem}
\label{thm:utility_additiveWishart}
Let $W \sim \W_{p+1}(\sigma^2 I, k)$, and denote $N \in \mathbb{R}^{p\times p}$ and $\vec n\in\mathbb{R}^p$ s.t. $ W = \left(\begin{array}{c c} N &\vec n\cr \vec n\T & \ast \end{array} \right)$. Let $X\in\mathbb{R}^{n\times p}$ be a matrix s.t. $X\T X$ is invertible and let $\vec y\in \mathbb{R}^n$, and such that there exists a $C\geq 2$ s.t. $\svmin(X\T X) = C\cdot \sigma^2(\sqrt k + \sqrt p + \sqrt{2\ln(4/\nu)})^2$. Denote $\widetilde{X\T X} = X\T X + N$ and $\widetilde{X\T \vec y} = X\T \vec y + \vec n$. Then
\[ \left\|\tvec\beta-\hvec\beta \right\| 
\leq  \frac 1 {C-1} \|\hvec\beta\| + \frac{\sigma^2(C-2)}{(C-1)\svmin(X\T X)} \cdot\min\left\{2\sqrt{2kp\cdot \ln(4p/\nu)},(\sqrt k + \sqrt p + \sqrt{2\ln(4/\nu)})^2\right\} \] 
\end{theorem}
\begin{proof}
Because $\sigma^2 I$ is a diagonal matrix, standard results on the Wishart distribution give that $N \sim \W_p(\sigma^2I_{p\times p},k)$. We therefore denote $R$ as a $(k\times p)$-matrix of i.i.d samples from a normal Gaussian $\gauss{0,1}$, and have $N = \sigma^2 R\T R$. The Woodbury formula gives that
\begin{align*}
(X\T X + N)^{-1} & = (X\T X)^{-1} - \sigma^2(X\T X)^{-1} R\T(I - \sigma^2 R (X\T X)^{-1} R\T)^{-1} R (X\T X)^{-1}
\intertext{Denoting $Q=\sigma R(X\T X)^{-1/2}$ we get}
& =(X\T X)^{-1} - (X\T X)^{-1/2}\left[ Q\T (I - QQ\T)^{-1} Q\right] (X\T X)^{-1/2}
\end{align*}
Now, if we denote $Q = U\Lambda V\T$ where $Q$'s singular values are $\lambda_1, \ldots, \lambda_d$, we get
 $ Q\T(I-QQ\T)^{-1} Q = V \cdot \textrm{diag}\left( \frac {\lambda_i^2}{1-\lambda_i^2} \right)_i \cdot V\T = V \cdot \textrm{diag}\left( \frac {1}{1-\lambda_i^{2}}-1 \right)_i \cdot V\T$. Note that $Q\T Q = \sigma^2 (X\T X)^{-1/2} R\T R (X\T X)^{-1/2}$ and so, due to Lemma~\ref{lem:eigenvalues_Wishart} we have $\lambda_1^2 = \sigma_{\max}(Q\T Q) \leq \frac{\sigma^2(\sqrt k + \sqrt p + \sqrt{2\ln(4/\nu)})^2}{\sigma_{\min}(X\T X)} \leq C^{-1}$ w.p. $\geq 1- \nu/2$. Which means that w.p. $\geq 1- \nu/2$ we have $\svmax( Q\T(I-QQ\T)^{-1} Q) \leq \tfrac 1 {C-1}$. And so we have that both (i)  $(X\T X)^{-1}-(X\T X+N)^{-1} \preceq \tfrac 1{C-1} (X\T X)^{-1}$ and (ii) $(X\T X+N)^{-1} \preceq \tfrac {C-2}{C-1} (X\T X)^{-1}$.

Next we turn to bound $\|\vec n\|$. One easy bound, given Lemma~\ref{lem:eigenvalues_Wishart}, is to show that w.p. $\geq 1- \nu/2$ it holds that \[\|\vec n\| \leq \|W\vec e_d\| \leq \|W\|\cdot 1 \leq \sigma^2(\sqrt k + \sqrt p + \sqrt{2\ln(4/\nu)})^2\] Alternatively we can derive the following bound. Each coordinate in $\vec n$ is the result of the dot-product between the $j$-th column of $R$, denoted $\vec r_j$ with the $d$-th column of $R$, denoted $\vec r_d$. Each coordinate in $R$ is sampled i.i.d from $\gauss{0,\sigma^2}$. Next, we use the fact that for two independent Gaussians \emph{with the same variance} $X,Y\sim \gauss{0,\sigma^2}$ it holds that
$XY = \tfrac {(X+Y)^2}2 - \tfrac{(X-Y)^2}2$ with $\tfrac 12 (X+Y)$ and $\tfrac1 2(X-Y)$ are two independent\footnote{This is where we need to use the fact that $X$ and $Y$ have the same variance. We have $\left(\begin{array}{c} X+Y \cr X-Y \end{array}\right) = \left(\begin{array}{c c} 1 & 1\cr 1 &-1 \end{array}\right) \left(\begin{array}{c} X \cr Y \end{array}\right)$ and so the variance of $\left(\begin{array}{c} X+Y \cr X-Y \end{array}\right)$ is diagonal iff $X$ and $Y$ have the same variance.} Gaussians $\gauss{0, \tfrac{\sigma^2}2}$. And so $\vec r_j\cdot \vec r_d = Z_{j_1} - Z_{j_2}$ where $Z_{j_1}, Z_{j_2} \sim \tfrac {\sigma}{\sqrt 2}\cdot  \chi_k^2$. Tail bounds for the $\chi^2$-distribution (see Claim~\ref{clm:chi_squared_tailbound}) give that w.p. $\geq 1-\nu/2$ it holds that each coordinate of $\vec n$ is bounded in absolute value by $\tfrac{\sigma^2}2(\sqrt k + \sqrt{2\ln(4p/\nu)})^2-\tfrac{\sigma^2}2(\sqrt k - \sqrt{2\ln(4p/\nu)})^2 =4\sqrt{2k\ln(4p/\nu)}$, which means $\|n\| \leq 2\sigma^2\sqrt{2k\cdot \ln(4p/\nu)}$.\footnote{We conjecture that the true bound in $\log(p)$-factor smaller, i.e. $O(\sigma^2\sqrt{2kp\cdot\ln(4/\nu)})$.}

Combining both bounds, we have that w.p. $\geq 1-\nu$ it holds that 
\begin{align*}
&&\hvec\beta - \tvec\beta &= \left((X\T X)^{-1} -(X\T X+N)^{-1}\right) X\T \vec y - (X\T X+N)^{-1}\vec n\cr
&&\Rightarrow \|\hvec\beta - \tvec\beta\| &\leq \frac 1 {C-1} \|(X\T X)^{-1}X\T y\| + \frac{2\sigma^2(C-2)}{C-1} \|(X\T X)^{-1}\|\sqrt{2kp\cdot \ln(4p/\nu)}\cr
&&& = \frac 1 {C-1} \|\hvec\beta\| + \frac{2\sigma^2(C-2)}{(C-1)\svmin(X\T X)} \sqrt{2kp\cdot \ln(4p/\nu)}\cr
&\textrm{or: } &\|\hvec\beta - \tvec\beta\| &\leq\frac 1 {C-1} \|\hvec\beta\| + \frac{\sigma^2(C-2)}{(C-1)\svmin(X\T X)} (\sqrt k + \sqrt p + \sqrt{2\ln(4/\nu)})^2\end{align*}
\end{proof}

\cut{
\newpage

\begin{theorem}
Fix $\epsilon>0$ and $\delta \in(0,\tfrac 1 e)$. Fix $B>0$. 
Fix a positive integer $r > 2\lndelta$ and denote \[w = B \cdot \max \left\{ \sqrt{\frac {2r + 2\sqrt{2 r \lndelta} + 2\lndelta} {\epsilon}}~,~~ 6\right\}\] Let $A$ be a $n\times d$ matrix where each row of $A$ has bounded $L_2$-norm of $B$. Given that $\sigma_{\min}(A) \geq w$, the algorithm that picks a $r\times n$ matrix $R$ whose entries are iid samples from a normal  distribution $\gauss{0,1}$ and publishes $R\cdot A$ is $(\epsilon,\delta)$-differentially private.
\end{theorem}
\begin{proof}
Fix $A$ and $A'$ be two neighboring $n\times d$ matrix, s.t. $A-A'$ is a rank-$1$ matix and denote $E = A-A' = e_i (\vec v-\vec v')\T$, and clearly $\sigma_{\min}(A),\sigma_{\min}(A') \geq w$ and $\|E\| = \|\vec v-\vec v'
\| \leq 2B$. We transpose $A$ and $R$ and denote $X = A\T R\T$ and $X' = (A')\T R\T$. For each column $j$ of $R\T$ it holds that $R\T_j \sim \gauss{\vec{0}, I_{n\times n}}$, and therefore the $j$-th column of $X$ is distributed like a random variable from $\gauss{\vec{0}, A\T A}$. Furthermore, as the columns of $R$ are independently chosen, so are the columns of $X$ are independent of one another. Therefore, for any $r$ vectors $\vec x_1, ...,\vec x_r\in \R^d$ it holds that
\[ \PDF_X( \vec x_1, ..., \vec x_r) = \prod_{j=1}^r \frac 1 {\sqrt{(2\pi)^d\det(A\T A)}} \exp\left(-\tfrac 1 2 \vec x_j\T (A\T A)^{-1} \vec x_j \right)\] therefore 
\begin{eqnarray*}
\frac{\PDF_X( \vec x_1, ..., \vec x_r)} {\PDF_{X'}( \vec x_1, ..., \vec x_r)} && = \prod_{j=1}^r \sqrt{\frac {\det(A'\T A')} {\det(A\T A)}} \exp\left(-\tfrac 1 2 \vec x_j\T ((A\T A)^{-1} - (A'\T A')^{-1}) \vec x_j \right)\cr 
&& = \left(\frac {\det(A'\T A')} {\det(A\T A)} \right)^{\tfrac r 2}\exp\left(-\tfrac 1 2\sum_j \vec x_j\T ((A\T A)^{-1} - (A'\T A')^{-1}) \vec x_j \right)
\end{eqnarray*}
Let $\sigma_1,\ldots, \sigma_d$ denote the singular values of $A$ and $\lambda_1, \ldots, \lambda_d$ denote the singular values of $A'$. 
We have that $A\T A - A'\T A' = \vec v \vec v\T - \vec v' \vec v'\T$ so $A-A'$ is a rank-$2$ matrix whose operator norm is $\leq 2B^2$. Using Weyl's theorem we have that $\left|\sum_i \sigma_i^2 - \sum_i \lambda_i^2\right| \leq 2B^2$ and so
\begin{align*}
&\left(\frac {\det(A'\T A')} {\det(A\T A)}\right)^{r/2} = \left(\prod_{i} \frac {\lambda_i^2}{\sigma_i^2}\right)^{r/2} = \prod_{i}\left(1+\frac {\lambda_i^2-\sigma_i^2}{\sigma_i^2}\right)^{r/2} \leq \exp(\frac r {2w^2}\sum_{i} \lambda_i^2-\sigma_i^2)\leq\exp(\frac {rB^2}{w^2})
\end{align*}

We now bound the second term in the $\PDF$. 
%
We denote $\bar A$ as the matrix we get by zeroing out the $i$-th row of $A$ or $A'$, meaning, $A = \bar A + e_i \vec v\T$ and $A' = \bar A + e_i \vec v'\T$. Therefore, $A\T A  = \bar A \T \bar A + \vec v \vec v\T$ and $A'\T A' = \bar A\T  \bar A + \vec v' \vec v'\T$. Using Sherman-Morrison formula $\left( (X+\vec u\vec v\T)^{-1} = X^{-1} - \frac {X^{-1} \vec u\vec v\T X^{-1}}{1+\vec v\T X^{-1} \vec u}  \right)$ we get
\begin{align*}
0_{d\times d} & = (\bar A\T \bar A)^{-1} - (\bar A\T \bar A)^{-1} \cr
& = \left( (A \T A)^{-1} + \frac{(A \T A)^{-1} \vec v \vec v\T (A \T A)^{-1} } {1 - \vec v\T (A \T A)^{-1} \vec v}\right) - \left(
 ( A' \T A')^{-1} + \frac{(A' \T A')^{-1} \vec v' \vec v'\T (A' \T A')^{-1} } {1 - \vec v'\T (A' \T A')^{-1} \vec v'}\right)
\end{align*}
It follows that for every $\vec x_j$ it holds that
\begin{align*}
0 &= \vec x_j\T \Big((A\T A)^{-1} - (A'\T A')^{-1}\Big) \vec x_j + \frac{ (\vec x_j\T (A\T A)^{-1}\vec v)^2}{1 - \vec v\T (A \T A)^{-1} \vec v} - \frac{ (\vec x_j\T (A'\T A')^{-1}\vec v')^2}{1 - \vec v'\T (A' \T A')^{-1} \vec v'}
\end{align*}
Denoting $c^2 = ({1-\vec v\T (A\T A)^{-1} \vec v})^{-1}$ and $c'^2 = ({1-\vec v'\T (A'\T A')^{-1} \vec v'})^{-1}$ (so $c^2,c'^2 \leq (1-\tfrac {B^2}{w^2})^{-1}$) we have that
\begin{align*}
& -\tfrac 1 2\left(\vec x_j\T ((A\T A)^{-1} - (A'\T A')^{-1}) \vec x_j\right) = \tfrac 1 2 c^2(\vec x_j\T (A\T A)^{-1}\vec v)^2- \tfrac 1 2c'^2(\vec x_j\T (A'\T A')^{-1}\vec v')^2 \cr
\Rightarrow &\left|-\tfrac 1 2\sum_{j=1}^r\left(\vec x_j\T ((A\T A)^{-1} - (A'\T A')^{-1}) \vec x_j\right)\right|  = \tfrac 1 2\left| c^2\sum_{j=1}^r(\vec x_j\T (A\T A)^{-1}\vec v)^2- c'^2\sum_{j=1}^r(\vec x_j\T (A'\T A')^{-1}\vec v')^2\right|
\cr & \leq \tfrac 1 2\max\left\{c^2\sum_{j=1}^r(\vec x_j\T (A\T A)^{-1}\vec v)^2, c'^2\sum_{j=1}^r(\vec x_j\T (A'\T A')^{-1}\vec v')^2 \right\}
\end{align*} where the last inequality holds since both terms are non-negative.

In the rest of the proof, our goal is to bound, w.h.p., the two non-negative terms $\sum_j(\vec v\T (A\T A)^{-1} \vec x_j)^2$ and $\sum_j(\vec v\T (A'\T A')^{-1} \vec x_j)^2$. We continue assuming $\vec x_j$ is sampled from $X_j = A\T R_j\T$, but the argument is symmetric in the case $\vec x_j$ is sampled from $X'_j = A'\T R_j\T$.

Recall that $R_j\T$ (the $j$-th row of $R$) is sample from a multivariate Gaussian $\gauss{\vec 0, I_{n\times n}}$. Therefore, $\vec v\T (A\T A)^{-1} \vec x_j = \vec v\T (A\T A)^{-1} A\T R_j \sim \gauss{0, \vec v\T (A\T A)^{-1} \vec v\T}$. Denoting $\psi_1^2 = \vec v\T (A\T A)^{-1} \vec v\T \leq \tfrac{B^2}{w^2}$, we have that the variance of the term $\vec v\T (A\T A)^{-1} \vec x_j$ is distributed like $\mathcal{N}(0, \psi_1^2)$.

As for the term $\vec v\T (A'\T A')^{-1} \vec x_j$, we have that
\begin{align*}
\vec v\T (A'\T A')^{-1} \vec x_j &= \vec v\T (A'\T A')^{-1} A\T R_j = \vec v\T (A'\T A')^{-1} (A'+E)\T R_j  
\end{align*}
so $\vec v\T (A'\T A')^{-1} \vec x_j \sim \gauss{0,~ \vec v\T (A'\T A')^{-1} (A'+E)\T(A'+E)(A'\T A')^{-1} \vec v}$. Simplifying, we get
\begin{align*}
&\vec v\T (A'\T A')^{-1} (A'+E)\T (A'+E) (A'\T A')^{-1} \vec v \cr&= \vec v\T (A'\T A')^{-1} \vec v + 2 \vec v\T (A'\T A')^{-1} E A'^\dag \vec v + (\vec v\T (A'\T A')^{-1} (\vec v-\vec v'))^2 \stackrel{\rm def}=\psi_2^2
\end{align*}
so $\vec v\T (A'\T A')^{-1} \vec x_j\sim \gauss{0, \psi_2^2}$, with $\psi_2^2 \leq \tfrac {B^2}{w^2} + 2 \tfrac {2B\cdot B^2}{w^3} + \tfrac {4B^4}{w^4} = \tfrac{B^2}{w^2}(1+\tfrac {2B} w)^2$. 

To summarize, we have that $\vec v\T (A\T A)^{-1} \vec x_j = \psi_1 Y_j$ and $\vec v'\T (A'\T A')^{-1} \vec x_j = \psi_2 Y_j$ where $Y_j \sim \gauss{0,1}$ . It follows that
\begin{align*}
& \sum_{j=1}^r (x_j\T (A\T A)^{-1} \vec v)^2 = \psi_1^2 \sum_{j=1}^r Y_j^2 \cr
& \sum_{j=1}^r (x_j\T (A'\T A')^{-1} \vec v')^2 = \psi_2^2 \sum_{j=1}^r Y_j^2 \cr
\end{align*}
All that is left is to use tail-bounds on $\chi_r^2$ distributions to bound w.h.p $\sum_{j=1}^r Y_j^2$.

So, using Claim~\ref{clm:upper_tailbound} and the fact that $c,c' \leq (1- \tfrac {B^2}{w^2})^{-1}$ we have that w.p. $\geq 1- \delta$ it holds that
\begin{align*}
-\tfrac 1 2 \sum_j \left(\vec x_j\T ((A\T A)^{-1} - (A'\T A')^{-1}) \vec x_j\right) &\leq   \tfrac 1 2 \max \{ c^2 \psi_1^2, c'^2 \psi_2^2\} (r + 2\sqrt {2r ( \lndelta}+2\lndelta)\cr
&\leq \frac {B^2}{w^2}\cdot \frac {(1+\tfrac {2B}{w})^2}{2(1-\tfrac {B^2}{w^2})}(r + 2\sqrt {2r ( \lndelta}+2\lndelta)
\cr & \stackrel {\tfrac B w<\tfrac 1 6}{\leq} ~~\frac {B^2}{w^2}(r + 2\sqrt {2r ( \lndelta}+2\lndelta)
\end{align*}
Altogether, setting $w \geq B \sqrt{\frac {2r + 2\sqrt {2r ( \lndelta}+2\lndelta} {\epsilon}}$ we have that
\begin{eqnarray*}
\frac{\PDF_X( \vec x_1, ..., \vec x_r)} {\PDF_{X'}( \vec x_1, ..., \vec x_r)} & \leq &\exp\left(\frac {B^2}{w^2} \cdot (r + r + 2\sqrt {2r ( \lndelta}+2\lndelta)\right)\cr
& \leq &\exp\left(\frac {B^2}{w^2} \cdot (2r + 2\sqrt {2r ( \lndelta}+2\lndelta\right) \leq e^\epsilon
\end{eqnarray*}
\end{proof}
Observe that $w = O(\frac B {\sqrt\epsilon}) \max\{r^{\tfrac 1 2},  r^{\tfrac 1 4}\sqrt{ \lndelta}\}$. This is quite surprising as we have an algorithm that guarantees $(\epsilon,\delta)$-differential privacy, that has no dependence on $\delta$ as long as $\delta > 2^{-\sqrt r}$.

}

\end{document}